\documentclass{CSML}

\pdfoutput=1

\usepackage{lastpage}
\def\dOi{13(3:26)2017}
\lmcsheading%
{\dOi}
{1--\pageref{LastPage}}
{}
{}
{Apr.~21, 2016}
{Sep.~26, 2017}
{}

\usepackage[utf8]{inputenc}
\usepackage[T1]{fontenc}
\usepackage{lmodern}
\usepackage{microtype}
\usepackage{amssymb}
\usepackage{amsmath}
\usepackage{amsthm}
\usepackage{thmtools}

\usepackage[ruled,vlined,linesnumbered,nokwfunc,algo2e]{algorithm2e_}

\usepackage{tikz}

\usepackage{xcolor}
\usepackage{xspace}
\usepackage{comment}

\usepackage[hidelinks]{hyperref} % Option ocgcolorlinks makes links only colored on 

\usetikzlibrary{arrows,shapes}
\usetikzlibrary{decorations.pathmorphing}
\tikzstyle{even}=[circle,draw,inner sep=2pt,thick,minimum size=7mm]
\tikzstyle{odd}=[regular polygon sides=3,regular polygon,draw,inner sep=1pt,thick,minimum size=4mm]
\tikzstyle{arrow}=[->,line width=1pt,>=stealth',bend angle=20]
\tikzstyle{label}=[draw,circle,fill,inner sep=1.2pt,solid]
\definecolor{lightgray}{rgb}{0.85,0.85,0.85}

\colorlet{DarkRed}{red!50!black}
\colorlet{DarkGreen}{green!50!black}
\colorlet{DarkBlue}{blue!50!black}

\declaretheorem[numberwithin=section]{theorem}
\declaretheorem[numberlike=theorem]{lemma}

\declaretheorem[numberlike=theorem]{corollary}

\declaretheorem[numberlike=theorem]{remark}
\declaretheorem[numberlike=theorem]{example}
\DontPrintSemicolon
\SetAlCapFnt{\normalfont\scshape}
\SetProcNameSty{texttt}
\SetFuncSty{texttt}

\usepackage{booktabs}

\newcommand{\set}[1]{\{#1\}}
\newcommand{\lu}{\textup{(}}
\newcommand{\ru}{\textup{)}}
\newcommand{\upbr}[1]{\lu #1\ru}

\newcommand{\Nats}{\mathbb{N}\xspace}

\newcommand{\sseq}{\langle v_0,v_1,v_2,\ldots\rangle}
\newcommand{\pat}{\omega\xspace} 

\newcommand{\obj}{\phi\xspace}

\newcommand{\idxs}{j\xspace}

\newcommand{\rev}{\textit{RevG}\xspace}

\newcommand{\scc}{C\xspace}
\newcommand{\inscc}{\expandafter\MakeLowercase\expandafter{\scc}\xspace} 
\DeclareMathOperator{\Out}{\textit{Out}\xspace}
\DeclareMathOperator{\In}{\textit{In}\xspace}
\DeclareMathOperator{\OutDeg}{\textit{Outdeg}\xspace}
\DeclareMathOperator{\InDeg}{\textit{Indeg}\xspace}

\newcommand{\bsccalg}{\ProcNameSty{SmallestBSCC}}
\newcommand{\reachG}[2]{\ProcNameSty{GraphReach}(#1, #2)}

\newcommand{\blue}{\textit{Bl}\xspace}
\newcommand{\gc}{X\xspace} 
\newcommand{\idxg}{i\xspace} 

\newcommand{\at}[3]{\textit{Attr}_{#1}(#2, #3)\xspace}
\newcommand{\idxa}{i\xspace}
\newcommand{\straa}{\sigma\xspace}

\newcommand{\strab}{\pi\xspace}

\newcommand{\pe}{{\mathcal{E}}\xspace}
\newcommand{\po}{{\mathcal{O}}\xspace}
\newcommand{\prio}{\alpha\xspace}
\newcommand{\numprio}{c\xspace}
\newcommand{\ve}{V_\pe\xspace}
\newcommand{\vo}{V_\po\xspace}

\newcommand{\we}{W_\pe\xspace}
\newcommand{\wo}{W_\po\xspace}
\newcommand{\pl}{{z}\xspace}
\newcommand{\op}{{\overline{z}}\xspace}
\newcommand{\pgame}{\mathcal{P}\xspace}
\newcommand{\game}{\mathcal{G}\xspace}
\newcommand{\domsize}{h\xspace}
\newcommand{\dombound}[2]{h\xspace} %{H(#1, #2) %\Gamma
\DeclareMathOperator{\domalg}{\ProcNameSty{Dominion}\xspace}

\DeclareMathOperator{\progress}{\ProcNameSty{ProgressMeasure}\xspace}

\newcommand{\streettalg}{{\textsc{Streett}}\xspace}

\newcommand{\remove}{\ProcNameSty{Remove}}
\newcommand{\bad}{\ProcNameSty{Bad}}
\newcommand{\construct}{\ProcNameSty{Construct}}
\newcommand{\bits}{\textit{bits}\xspace}
\newcommand{\ds}{\mathit{D}\xspace}
\newcommand{\liste}{Q\xspace}

\hypersetup{
	pdftitle = {Improved Algorithms for Parity and Streett objectives},
	pdfauthor = {Krishnendu Chatterjee, Monika Henzinger, Veronika Loitzenbauer}
}

\begin{document}

\title[Improved Algorithms for Parity and Streett objectives]
      {Improved Algorithms for Parity and Streett objectives\rsuper*}

\author[K.~Chatterjee]{Krishnendu Chatterjee\rsuper a}
\address{{\lsuper a}IST Austria, Klosterneuburg, Austria}
\email{krishnendu.chatterjee@ist.ac.at}
\thanks{{\lsuper{a,b,c}}The authors are partially supported by the
  Vienna Science and Technology 
Fund (WWTF) grant ICT15-003 and the Austrian Science Fund (FWF): P23499-N23.}
\thanks{{\lsuper a}Partially 
supported by the Austrian Science Fund (FWF): S11407-N23 (RiSE/SHiNE), 
an ERC Start Grant (279307: Graph Games), and a Microsoft Faculty Fellows Award.}

\author[M.~Henzinger]{Monika Henzinger\rsuper b}
\address{{\lsuper b}University of Vienna, Faculty of Computer Science, Vienna, Austria}
\email{monika.henzinger@univie.ac.at}
\thanks{{\lsuper{b,c}}The research leading to these results has received funding from the European 
Research Council under the European Union's Seventh Framework Programme 
(FP/2007-2013) / ERC Grant Agreement no. 340506 and the Vienna Science and 
Technology Fund (WWTF) grant ICT10-002.}

\author[V.~Loitenbauer]{Veronika Loitzenbauer\rsuper c}
\address{{\lsuper c}Bar-Ilan University, Ramat Gan, Israel}
\email{veronika@datalab.cs.biu.ac.il}
\thanks{{\lsuper c}Work done while at University of Vienna.}

\titlecomment{{\lsuper*}A preliminary version of this paper appeared in~\cite{ChatterjeeHL15}.}

% \date{}

\begin{abstract}
\noindent The computation of the winning set for parity objectives and for 
Streett objectives in graphs as well as in game graphs
are central problems in computer-aided verification, with application to 
the verification of closed systems with strong fairness conditions,
the verification of open systems, checking interface compatibility, 
well-formedness of specifications, and the synthesis of reactive systems. 
We show how to compute the winning set on $n$ vertices 
for (1)~parity-3 (aka one-pair Streett) objectives in game graphs
in time~$O(n^{5/2})$
and for (2)~k-pair Streett objectives in graphs
in time~$O(n^2 + nk \log n)$. 
For both problems this gives faster algorithms for dense graphs and 
represents the first improvement in asymptotic running 
time in 15 years. 
\end{abstract}

\maketitle

\section{Introduction}

In the formal verification and synthesis of systems \emph{graphs} and \emph{game graphs}
are fundamental models of systems, where vertices correspond to states
of the systems and edges correspond to transitions between states. \emph{$\omega$-regular
objectives} are a canonical way to specify desired and undesired behaviors
of systems, and the algorithmic questions are to determine whether a model
satisfies its specification or to generate a strategy to satisfy the specification.
In this work we study \emph{Streett} and \emph{parity} objectives
that can express all $\omega$-regular objectives and the algorithmic 
questions we consider are therefore core questions in formal verification and synthesis.
We first define the problem, next discuss its significance and previous work,
and then present our contributions.

\smallskip\noindent\emph{Game graphs and graphs.}
Consider a directed graph $(V,E)$ with a partition $(V_1, V_2)$ of $V$, which is
called a \emph{game graph}. Let $n = |V|$ and $m = |E|$. 
On the graph two players play the following \emph{alternating game}
that forms an infinite path. They start by placing a token on an initial 
vertex and then take turns indefinitely in moving the token: At a vertex 
$v \in V_1$ player 1 moves the token along one of the outedges of~$v$, at a 
vertex $u \in V_2$ player 2 moves the token along one of the outedges of~$u$. 
If $V_2=\emptyset$, then we simply have a standard graph.

\smallskip\noindent\emph{Objectives and winning sets.} Objectives are subsets of infinite paths
that specify the desired set of paths for player~1, and the objective for player~2 is the 
complement of the player-1 objective (i.e., we consider \emph{zero-sum} games). Given an objective $\obj$, 
an infinite path \emph{satisfies the objective} if it belongs to $\obj$.
Given a starting vertex $x \in V$ and an objective $\obj$, if player~1 can guarantee that
the infinite path starting at $x$ satisfies $\obj$,  \emph{no matter what choices player~2 makes,} 
then player~1 can \emph{win from $x$} and $x$ belongs to the \emph{winning set of player~1}. 
The winning sets partition the game graph, i.e., the complement of the winning 
set for player~1 is the winning set for player~2.
In case the game graph is a standard graph (i.e., $V_2=\emptyset$), the \emph{winning set} 
consists of those vertices $x$ such that there exists an infinite path starting at $x$ 
that satisfies $\obj$.
%%The decision problem given a vertex $x$ asks whether $x$ belongs to the winning set. 
The winning set computation %decision problem 
for game graphs is more involved than for standard graphs due to 
the presence of the adversarial player~2.

\smallskip\noindent\emph{Relevant objectives.}
The most basic objective is \emph{reachability} where, given a set $U \subseteq V$
of vertices, an infinite path satisfies the objective if the path visits a vertex
of $U$ \emph{at least once}.
The next interesting objective is the \emph{Büchi} objective that requires an 
infinite path to visit some vertex of $U$ \emph{infinitely often}.
The next and a very central objective in formal verification and automata theory is the 
\emph{one-pair Streett objective} that consists of a pair $(L_1,U_1)$ of sets of vertices 
(i.e., $L_1 \subseteq V$ and $U_1 \subseteq V$), and an infinite path satisfies the objective 
iff the following condition holds:
if some vertex of $L_1$ is visited infinitely often, then some vertex of $U_1$ is visited
infinitely often (intuitively the objective specifies that if one Büchi objective holds,
then another Büchi objective must also hold).
A generalization of one-pair Streett objectives is the \emph{$k$-pair Streett objective}
(aka \emph{general Streett objective})
that consists of $k$-Streett pairs $(L_1,U_1), (L_2,U_2), \ldots, (L_k,U_k)$, and 
an infinite path satisfies the objective iff the condition for every Streett pair is 
satisfied (in other words, the objective is the conjunction of $k$ one-pair Streett objectives).
A different generalization of one-pair Streett objectives are \emph{parity objectives}.
For a parity objective the input additionally contains a priority function that 
assigns each vertex a natural number called priority. The parity objective is 
satisfied if the highest priority visited infinitely often is even. 
Parity objectives with at most 2~different priorities are equivalent to Büchi objectives
and parity objectives with at most 
3~different priorities are equivalent to one-pair Streett objectives.

We study %two problems which are core algorithmic questions in verification:
(1)~game graphs with parity-3 (aka one-pair Streett) objectives and their generalization 
to parity objectives, and (2)~graphs with general Streett objectives.

\smallskip\noindent\emph{Significance in verification.} 
Two-player games on graphs %played by player~1 and the adversary player~2 
are useful in many problems in computer science,
specially in the verification and synthesis of systems such as the 
%%synthesis of systems from specifications and 
synthesis of reactive systems~\cite{Church62,PnueliR89,RamadgeW87},  
the verification of open systems~\cite{AlurHK02}, and checking interface 
compatibility~\cite{InterfaceAutomata} and the well-formedness of 
specifications~\cite{Dill89book}, and many others. 
General and one-pair Streett objectives are central in verification as most 
commonly used specifications can be expressed as Streett 
automata~\cite{Safra88,Thomas97}. 
Parity objectives are also canonical to express properties in verification, 
since every Streett automaton can be converted to a parity automaton~\cite{Safra92}. 
Moreover, parity objectives are particularly important 
as solving parity games is equivalent to $\mu$-calculus model checking~\cite{EmersonJ91}.
Dense game graphs can emerge from a synchronous product of several components 
(where each component makes transitions at each step)~\cite{KuijperP09,ChatterjeeGIP16}.

\emph{Game graphs with parity-3 objectives} arise in many applications in 
verification. We sketch a few of them.
(A)~Timed automaton games are a model for real-time systems.
The analysis of such games with reachability and safety objectives 
(which are the dual of reachability objectives) reduces to 
game graphs with parity-3 objectives~\cite{deAlfaroFHMS03,deAlfaroF07,ChatterjeeHP11,ChatterjeeP13}. 
(B)~The synthesis of Generalized Reactivity(1) (aka GR(1)) specifications exactly 
require the solution of game graphs with parity-3 objectives~\cite{BloemCGHJ10};
GR(1) specifications are standard for hardware synthesis~\cite{PitermanPS06} and even used in 
the synthesis of industrial protocols~\cite{GodhalCH13,BloemJPPS12}\footnote{A GR(1) specification 
expresses that if a conjunction of Büchi objectives holds, then another 
conjunction of Büchi objectives must also hold, and since conjunction of 
Büchi objectives can be reduced in linear time to a single Büchi 
objective, a GR(1) specification reduces to implication between two 
Büchi objectives, which is an parity-3 objective.}.
(C)~Finally, the problem of fair simulation~\cite{HenzingerKR02} between two systems also 
reduces to game graphs with parity-3 objectives~\cite{EtessamiWS05,ChatterjeeCK12}.

\emph{General Streett objectives in standard graphs} arise, for example, in the 
verification of closed systems with strong fairness 
conditions~\cite{LatvalaH00,DuretLutzPC09,Francez86}.
In program verification,
a scheduler is \emph{strongly fair} if every event that is enabled infinitely often 
is scheduled infinitely often. Thus,
verification of systems with strong fairness conditions directly corresponds to checking 
the non-emptiness of Streett automata, which in turn corresponds to determining the winning
set in standard graphs with Streett objectives. Note, however, that
a Streett objective can either specify desired behaviors of the system or 
erroneous ones, and for erroneous specifications, it is useful to have 
a \emph{certificate} (as defined below) to identify an error trace of the system~\cite{Ehlers10,LatvalaH00,DuretLutzPC09}, such as in the counterexample-guided 
abstraction refinement approach (CEGAR)~\cite{ClarkeGJLV03}.
%witness = error trace = example run of the system in which the error occurs
%explicit path given as a lasso

Note that \emph{standard graphs} are relevant for the verification of 
\emph{closed} systems or \emph{open} systems with demonic non-determinism 
(e.g., all inputs are from the environment that are not controllable); 
while \emph{game graphs} are relevant for the synthesis and 
verification of \emph{open} systems with both angelic and demonic 
non-determinism (e.g., certain inputs are controllable, and certain inputs
are not controllable).

\smallskip\noindent\emph{Previous results.} 
We summarize the previous results for game graphs and graphs with Streett
and Parity objectives. 

\emph{Game graphs.} We consider the computation of the winning set for player~1 in game graphs.
For \emph{reachability} objectives, the problem is PTIME-complete, and the computation can be 
achieved in time linear in the size of the graph~\cite{Beeri80,Immerman81}.
For \emph{Büchi} objectives, the current best known algorithm requires $O(n^2)$ time~\cite{ChatterjeeH14}.
For \emph{general Streett} objectives, the problem is coNP-complete~\cite{EmersonJ88}, and for \emph{one-pair 
Streett} objectives the current best known algorithm requires $O(m n)$ time~\cite{Jurdzinski00}.  
One-pair Streett objectives also corresponds to the well-known parity games problem 
with three priorities. Despite the importance of game graphs with parity-3 objectives in numerous applications and several algorithmic ideas to improve the 
running time for general parity games~\cite{VogeJ00,JurdzinskiPZ08,Schewe17} or Büchi games~\cite{ChatterjeeJH03,ChatterjeeH14}, 
there has been no algorithmic improvement since 2000~\cite{Jurdzinski00} for parity-3 games. The parity games problem in general is in UP $\cap$ coUP~\cite{Jurdzinski98}; 
it is one of the rare and intriguing combinatorial problems that lie in NP $\cap$ 
coNP but are not known to be in PTIME. 
Parity games can be solved by a randomized algorithm in time 
$n^{O(\sqrt{n/\log n})}$~\cite{BjorklundSV03} and
by deterministic algorithms in time $n^{O(\sqrt{n})}$~\cite{JurdzinskiPZ08}
and for~$\numprio\ge 3$ priorities in time $O(m \cdot (\kappa n / c^2)^{\gamma(\numprio)})$~\cite{Schewe17}
for a small constant $\kappa$ and $\numprio/3 \le \gamma(\numprio) \le \numprio/3 
+ 1/2$. Subsequent to our work, a quasi-polynomial time algorithm 
for parity games was achieved in a breakthrough result by 
Calude et.\ al~\cite{CaludeJKLS17,GimbertI17}. In follow-up work different
quasi-polynomial time algorithms as well as $O(m n \log(n)^{\numprio-1})$ time
algorithms for constant $\numprio$ were shown \cite{JurdzinskiL17,FearnleyJSSW17}.

\emph{Graphs.} In graphs the winning set for parity objectives with 
$\numprio$ priorities can be computed in $O(m \log{\numprio})$ 
time~\cite{ChatterjeeH11}. We study the computation of the winning 
set for general Streett objectives. If $x$ belongs to the winning set, it is 
often useful to output a \emph{certificate} for $x$. Let $S$ be a vertex set 
reachable from $x$ that induces a strongly connected subgraph such that for all 
$1 \leq \idxs \leq k$ we have either $S \cap L_\idxs=\emptyset$ or 
$S \cap U_\idxs \neq \emptyset$ (i.e., if $S$ contains a vertex from $L_\idxs$ 
then it also contains a vertex from $U_\idxs$). 
A certificate is a ``lasso-shaped'' path that consists of a path to $S$ and a 
(not necessarily simple) cycle between the vertices of~$S$~\cite{Ehlers10}. 
The basic algorithm~\cite{EmersonL87,LichtensteinP85} for the winning set problem 
has an asymptotic running time of $O((m+b)\min(n,k))$ with 
$b = \sum_{\idxs=1}^k (\lvert L_\idxs \rvert + \lvert U_\idxs \rvert) \le 2nk$. 
Within the same time bound Latvala and Heljanko~\cite{LatvalaH00} additionally
compute a certificate of size at most $n\min(n,2k)$.
Duret-Lutz~et~al.~\cite{DuretLutzPC09} presented a space-saving ``on-the-fly'' 
algorithm with the same time complexity for the slightly different transition-based 
Streett automata.
The current fastest algorithm for the problem by Henzinger and Telle~\cite{HenzingerT96} 
from 1996
has a running time of $O(m \min(\sqrt{m \log n}, k, n) + b \min(\log n, k))$;  
however, given a start vertex $x$, to report the certificate for $x$
adds an additive term of $O(n\min(n,k))$ to the running time bound.

%\end{compactenum}

\smallskip\noindent\emph{Our contributions.} In this work our contributions are two-fold.
%\begin{compactenum}

\emph{Game graphs.} We show that the winning set computation for game graphs with 
parity-3 objectives can be achieved in $O(n^{5/2})$ time.
Our algorithm is faster for $m \in \Omega( n^{1.5})$, and breaks the long-standing 
$O(m n)$ barrier for dense graphs.
Our algorithm for parity-3 games also extends to general parity games 
and improves the running time for dense graphs when the number of priorities 
is sub-polynomial in $n$. Let, as in~\cite{Schewe17},
$\gamma(\numprio) = \numprio/3 + 1/2 - 
4/(\numprio^2 - 1)$ for odd $\numprio$ and $\gamma(\numprio) = 
\numprio/3 + 1/2 - 1/(3 \numprio) - 4/\numprio^2$ for even $\numprio$, and let
$\beta(\numprio) = \gamma(\numprio)/(\lfloor\numprio / 2\rfloor + 1)$.
We obtain that the running time of our algorithm is 
$O(n^{1+\gamma(\numprio+1)}) = O(n^{2 + \gamma(\numprio) - \beta(\numprio)})$ 
for parity games with $\numprio$ priorities, i.e., for a constant number 
of priorities we replace $m$ of~\cite{Schewe17} by $n^{2-\beta(\numprio)}$.
Since the value of $\beta(\numprio)$ quickly approaches $2/3$ with
increasing~$\numprio$, we have that $n^{2-\beta(\numprio)}$ approaches $n^{4/3}$.
For small $\numprio$ we compare our running times with the Big-Step algorithm 
of~\cite{Schewe17} in Table~\ref{tab:comparison}.
\begin{table*}[!t]
\renewcommand{\arraystretch}{1.3}
\caption{Comparison of Running Times for Few Priorities.}\label{tab:comparison}
\centering
\begin{tabular}{@{}lccccc@{}}
\toprule
& \multicolumn{5}{c}{\# priorities}\\
\cmidrule{2-6}
 & 3 & 4 & 5 & 6 & 7 \\
\midrule
% $\gamma(\numprio)$ & $1$ & $3/2$ & $2$ & $7/3$ & $11/4$\\
% $\beta(\numprio)$ & $1/2$ & $1/2$ & $2/3$ & $7 / 12$ & $11/16$ \\
% $2-\beta(\numprio)$ & $3/2$ & $3/2$ & $4/3$ & $17 / 12$ & $21/16$ \\
Big-Step~\cite{Schewe17} & 
$O(m n)$ & $O(m n^{3/2})$ & $O(m n^{2})$ & $O(m n^{7/3})$ & $O(m n^{11/4})$\\

Big-Step~\cite{Schewe17} with $m = \Theta(n^2)$ & 
$O(n^3)$ & $O(n^{7/2})$ & $O(n^{4})$ & $O(n^{13/3})$ & $O(n^{19/4})$\\

Our algorithm & 
$O(n^{5/2})$ & $O(n^{3})$ & $O(n^{10/3})$ & $O(n^{15/4})$ & $O(n^{65/16})$\\
\bottomrule
\end{tabular}
\end{table*}

%%%%%%%%%%%

\emph{Graphs.}
We present an algorithm with $O(n^2 + b \log n)$ running time for the 
winning set computation in graphs with general Streett objectives, 
which is faster for $m \in \Omega( \max(n^{4/3} \log^{-1/3} n, b^{2/3} \log^{1/3} n))$ 
and $k \in \Omega(n^2 / m)$.
%Combined with the algorithm of Henzinger and Telle~\cite{HT1996} this yields a 
%running time of $O(\min(n^2, m \sqrt{m \log n}, mk) + b \min(\log n,k))$.
We additionally provide an algorithm that, given the winning set,
computes a certificate %%from this SCC 
for a vertex $x$ in the winning set in time $O(m + n  \min(n, k))$.
% in Section~\ref{sec:pre}. 
We also provide an example where the smallest certificate
has size $\Theta(n  \min(n, k))$, showing that no algorithm can 
compute and \emph{output} a certificate faster.
In contrast to~\cite{HenzingerT96}, the running time of our algorithm for the winning 
set computation does not change with certificate reporting.
Thus when certificates need to be reported and $k=\Omega(n)$, our algorithm is \emph{optimal} 
up to a factor of $\log n$ as the size of the input is at least $b$ and 
the size of the output is $\Omega(n^2)$.
%\end{compactenum}

\smallskip\noindent\emph{Technical contributions.}
% We now describe our main technical contributions. 
Both of our algorithms use a \emph{hierarchical \upbr{game} graph 
decomposition} technique that 
%originates in a graph decomposition technique 
was developed 
by Henzinger et al.~\cite{HenzingerKW99} 
%to find a new 
%\emph{connected component} 
to handle \emph{edge deletions} in 
\emph{undirected} graphs. In~\cite{ChatterjeeH14} it was extended 
to deal with \emph{vertex deletions} in \emph{directed}
and \emph{game} graphs. We combine and extend this technique in two ways.

%\begin{compactenum}
\emph{Game graphs.} The classical algorithm for parity-3 objectives repeatedly solves Büchi games such that the union of the winning sets of player~2 in the Büchi games is exactly the winning set for the parity-3 objective. 
Schewe~\cite{Schewe17} showed with his Big-Step algorithm that the small 
progress measure algorithm for parity games by 
Jurdzi\'nski~\cite{Jurdzinski00} can be used to compute small subsets of the 
winning set of player~2, called \emph{dominions}, and thereby improved the 
running time for general parity games. However, the Big-Step algorithm
does not improve the running time for parity-3 games. 
With Schewe's approach dominions with at 
most $\domsize$~vertices in Büchi games can be found in time $O(m \domsize)$. 
We extend this approach by using the hierarchical game graph decomposition 
technique to find small dominions quickly and call the $O(n^2)$ Büchi game 
algorithm of~\cite{ChatterjeeH14} for large dominions. 
This extension is possible as we are able to show that, rather surprisingly, 
it is sufficient to consider game graphs with $O(n \domsize)$ edges to detect 
dominions of size $\domsize$. Our approach extends to general parity games.

\emph{Graphs.} In prior work that used the hierarchical graph 
decomposition technique the running time analysis relied on the fact that
identified vertex sets that fulfilled a certain desired condition
were \emph{removed} from 
the (game) graph after their detection. The work for identifying the vertex set was 
then charged in an amortization argument to the
removed vertex set. This is not possible for general Streett objectives on graphs, 
where a strongly connected subgraph is identified and some but not all of 
its vertices might be removed. As a consequence a vertex might belong to 
an identified strongly connected subgraph multiple times. We 
show how to overcome this difficulty by identifying a strongly
connected subgraph with at most half of the vertices whenever a vertex set 
seizes to be strongly connected. We identify these strongly connected
subgraphs by running Tarjan's SCC algorithm~\cite{Tarjan72} on the graph and its 
\emph{reverse graph}, thereby finding the smallest \emph{top} (i.e. with no incoming edges) or
\emph{bottom} (i.e. with no outgoing edges) SCC contained in the formerly strongly
connected subgraph. The algorithm takes $O(n)$ time per vertex in the identified 
set. This will allow us to bound the total running time for this part of the
algorithm with~$O(n^2)$.

%\end{compactenum}
In Section~\ref{sec:parity} we present our algorithm for game graphs with 
parity objectives with $\numprio$ priorities, where
the special case of $\numprio = 3$ corresponds to one-pair Streett
objectives. In Section~\ref{sec:streett} we present the algorithm
for general Streett objectives in graphs.

\section{Parity Objectives in Game Graphs}\label{sec:parity}

\subsection{Preliminaries}

\smallskip\noindent\emph{Notation.}
Let for all $\numprio \in \Nats$ denote the set $\set{0, 1, \dots, \numprio - 1}$
by $[\numprio]$. %All logarithms are base~2.

\smallskip\noindent\emph{Parity games.}
A \emph{parity game} $\pgame = (\game, \prio)$ with $\numprio \le n$ priorities
consists of a \emph{game graph} $\game = (G, (\ve, \vo))$ with $G = (V, E)$
and a \emph{priority function} $\prio: V \rightarrow 
[\numprio]$ that assigns an integer from the set $[c]$ 
to each vertex. 
The sets $\ve \subseteq V$ and $\vo \subseteq V$ form a partition of $V$.
We denote the two players by $\pe$ (for even) and $\po$ (for odd).
We say that the vertices in $\ve$ are \emph{$\pe$}-vertices
and the vertices in $\vo$ are \emph{$\po$}-vertices.
Player~$\pe$ (resp.\ player~$\po$) wins a play if the \emph{highest} priority 
occurring infinitely often in the play is even (resp.\ odd).
 We use $\pl$ to denote 
one of the players $\set{\pe,\po}$ and $\op$ to denote her opponent.
\emph{Parity-3 games} are parity games with $\prio: V 
\rightarrow \set{0,1,2}$ and \emph{B\"uchi games} have $\prio: V \rightarrow \set{0,1}$, 
where the vertices in the set $B = \set{v \mid \prio(v) = 1}$ are called \emph{B\"uchi vertices}. %B\"uchi games are denoted as $(\game, B)$.

\begin{figure}
\centering
\begin{tikzpicture}
\matrix[column sep=8mm, row sep=8mm]{
	\node[even] (1) {$0$};
	& \node[odd] (2) {$0$};
	& \node[even] (3) {$2$};
	& \node[odd] (4) {$1$};
	& \node[odd] (5) {$1$};\\
	\node[even] (6) {$1$};
	& \node[odd] (7) {$1$};
	& \node[odd] (8) {$0$};
	& \node[odd] (9) {$2$};
	& \node[even] (10) {$1$};\\
};
\path (1) edge[arrow] (2)
					edge[arrow] (7)
					edge[arrow] (8)
					edge[arrow, bend left] (6)
			(6) edge[arrow, bend left] (1)
			(2) edge[arrow, bend left] (7)
					edge[arrow] (3)
			(7) edge[arrow] (6)
			(8) edge[arrow] (2)
			(9) edge[arrow] (8)
					edge[arrow, bend left] (3)
			(3) edge[arrow, bend left] (9)
					edge[arrow, bend left] (4)
					edge[arrow, bend left] (5)
			(4) edge[arrow, bend left] (3)
			(5) edge[arrow, bend left] (10)
			(10) edge[arrow, bend left] (5)
					edge[arrow] (9);
\end{tikzpicture}
\caption{An example of a parity game $\pgame = (\game, \prio)$ with three priorities. 
Circles denote $\pe$-vertices, triangles denote $\po$-vertices. 
The values in the vertices are the priorities.}\label{fig:game}
\end{figure}

\smallskip\noindent\emph{One-pair Streett and parity-3 games.}
A one-pair Streett objective with pair $(L_1, U_1)$ is equivalent to a parity 
game with three priorities. 
Let the vertices in $U_1$ have priority~$2$, let the vertices
in $L_1 \setminus U_1$ have priority~$1$, and let the remaining vertices 
have priority~$0$. Then player~1 wins the game with the one-pair Streett 
objective if and only if player~$\pe$ wins the parity-3 game. 
As our algorithm for parity-3 games extends to
general parity games, we use the notion of parity games (i.e., 
player~$\pe$ and player~$\po$ instead of player~1 and player~2).

\smallskip\noindent\emph{Plays.}
For technical convenience we consider that every vertex of the game graph 
$\game$ has at least one outgoing edge. 
A game is initialized by placing a token on a vertex. Then the two
players form an infinite path called \emph{play} in the game graph by 
moving the token along the edges. Whenever the token is on a vertex of~$V_\pl$, 
player~$\pl$ moves the token along one of the outgoing edges of the vertex.
Formally, a \emph{play} is an infinite sequence $\sseq$ of vertices such that
$(v_\ell,v_{\ell+1}) \in E$ for all $\ell \geq 0$. 

%%%
For a vertex $u\in V$, we write $\Out(G, u)=\set{v\in V \mid (u,v) \in E}$ 
for the set of successor vertices of~$u$ in $G$ 
and $\In(G, u)=\set{v \in V \mid (v,u) \in E}$ for the set of predecessor 
vertices of~$u$ in $G$. We denote by $\OutDeg(G, u)=|\Out(G, u)|$ 
the number of outgoing edges from~$u$, and by $\InDeg(G, u)=|\In(G, u)|$ 
the number of incoming edges. We omit the reference to $G$ if it is clear from 
the context.
%%%

\smallskip\noindent\emph{Strategies.}
A \emph{strategy} of a player~$\pl \in \set{\pe,\po}$ is a 
function that, given a finite prefix of a play ending at $v \in V_\pl$, 
selects a vertex from $\Out(v)$ to extend the play. 
\emph{Memoryless strategies} do not depend on the history of a play but only on the current vertex. That is, a memoryless 
strategy of player~$\pl$ is a function $\straa: V_\pl \rightarrow V$ such that for 
all $v \in V_\pl$ we have $\straa(v) \in \Out(v)$. It is well-known that for parity games it is sufficient
to consider memoryless strategies (see Theorem~\ref{thm:determinacy} below). 
Therefore we only consider memoryless strategies from now on.
A start vertex~$v$, a strategy~$\straa$ for $\pe$, and a 
 strategy~$\strab$ for $\po$ define a unique play $\pat(v,\straa,
\strab)=\sseq$ as follows: 
$v_0=v$ and for all $j \geq 0$, if $v_j \in \ve$, then $\straa(v_j)=v_{j+1}$, 
and if $v_j \in \vo$, then $\strab(v_j)=v_{j+1}$.
%%%

\smallskip\noindent\emph{Winning strategies and sets.}
A strategy~$\straa$ is winning for player~$\pl$ at start vertex~$v$ if the 
resulting play is winning for player~$\pl$ irrespective of the strategy of 
player~$\op$. A vertex~$v$ belongs to 
the \emph{winning set} $W_\pl$ of player~$\pl$ if player~$\pl$ has a winning 
strategy from $v$. By the following theorem every vertex is winning for exactly 
one of the two players. 
When required for explicit reference of a specific game graph~$\game$ or 
specific parity game~$\pgame$, we use $W_\pl(\game)$ and $W_\pl(\pgame)$ to refer to the 
winning sets. The algorithmic question for parity games is to compute the 
sets~$\we$ and~$\wo$. 
\begin{theorem}[\cite{EmersonJ91,McNaughton93}]\label{thm:determinacy}
For every parity game the vertices $V$ can be partitioned into the winning 
set~$\we$ of player~$\pe$ and the winning set~$\wo$ of player~$\po$. 
Each player has a memoryless strategy that is winning for her from 
all vertices in her winning set.
\end{theorem}

For the analysis of our algorithm we further introduce the notions of \emph{traps}, \emph{attractors}, and \emph{dominions}.

\smallskip\noindent\emph{Traps.}
A set $U \subseteq V$ is a \emph{$\pl$-trap} if for all $\pl$-vertices 
$u$ in $U$ we have $\Out(u) \subseteq U$ and for all $\op$-vertices 
$v$ in $U$ there exists a vertex $w \in \Out(v) \cap U$~\cite{Zielonka98}.
Note that player~$\op$ can ensure that a play that currently ends in a $\pl$-trap 
never leaves the $\pl$-trap against any strategy of player~$\pl$ by choosing an 
edge $(v,w)$ with $w \in \Out(v) \cap U$ whenever the current 
vertex~$v$ is in $U \cap V_\op$.
Given a game graph~$\game$ and a $\pl$-trap~$U$, we denote by $\game[U]$ 
the game graph induced by the set of vertices~$U$. Note that given that in 
$\game$ each vertex has at least one outgoing edge, the same property holds 
for $\game[U]$. By a slight abuse of notation,
we denote the sub-game induced
by $U$ by $(\game[U], \prio)$, where the priority function~$\prio$
is evaluated only on~$U$ and we say that the highest priority of $(\game[U], \prio)$
is $\max_{v \in U} \prio(v)$.

\smallskip\noindent\emph{Attractors.}
In a game graph $\game$, a $\pl$-\emph{attractor} $\at{\pl}{\game}{U}$ of a set 
$U \subseteq V$ is the set of vertices from which player~$\pl$ has a strategy 
to reach $U$ against all strategies of player~$\op$. We have that $U \subseteq 
\at{\pl}{\game}{U}$. 
A $\pl$-attractor can be constructed inductively as follows: Let $Z_0=U$; and for all $\idxa \ge 0$ let 
\begin{equation}\label{eq:attr}
\tag{\ddag}
\begin{split}
	Z_{\idxa+1} = Z_\idxa &\cup \set{v \in V_\pl \mid \Out(v) \cap Z_\idxa \neq \emptyset} \\
    &\cup \set{v \in V_\op \mid \Out(v) \subseteq Z_\idxa}.
\end{split}
\end{equation}
Then $\at{\pl}{\game}{U}= \bigcup_{\idxa\ge 0} Z_\idxa$.

The following lemma summarizes two well-known facts about attractors and 
traps that we use frequently.
\begin{lemma}\label{lem:attr}
Let $U \subseteq V$.
\begin{enumerate}
\item \cite{Beeri80,Immerman81} The attractor $A = \at{\pl}{\game}{U}$ can be computed in 
$O(\sum_{v \in A} \lvert\In(v)\rvert)$ 
time.\label{sublem:attrtime}
\item \cite[Lemma~4]{Zielonka98} Let~$\game$ be a game graph in which each vertex has at least one outgoing 
edge. 
Then the set $V \setminus \at{\pl}{\game}{U}$ is a $\pl$-trap in~$\game$.\label{sublem:complattr}
% \item Let $U$ be a $\pl$-trap in $\game$. Then $\at{\op}{\game}{U}$ is a $\pl$-trap~\cite[Lemma~5]{Zielonka98}.\label{sublem:attrclosed}
\end{enumerate}
\end{lemma}

\smallskip\noindent\emph{Dominions.}
A non-empty set of vertices $D \subseteq V$ is a 
$\pl$-\emph{dominion} if player~$\pl$ has a winning strategy from every vertex 
of $D$ that also ensures only vertices of $D$ are visited. Note that a $\pl$-dominion
is also a $\op$-trap and that the $\pl$-attractor of a $\pl$-dominion
is again a $\pl$-dominion. 
In parity games dominions of size $\lvert D \rvert \le \domsize + 1$ can be 
computed by running the small-progress measure algorithm of
Jurdzi\'nski~\cite{Jurdzinski00} with a reduced codomain~\cite{Schewe17}.
% A description of the small-progress measure algorithm for B\"uchi games is given 
% in~\cite[Section~2.4.1]{CH14}.
We use this algorithm, whose properties are stated in the lemma below, as a subroutine.
\begin{lemma}[\cite{Schewe17}]\label{lem:progress}
Let $(\game, \prio)$ be a parity game with game graph $\game$, 
$n$~vertices, $m$~edges, a priority function~$\prio$, and $\numprio$ priorities.
There is an algorithm $\progress(\game, \prio, \domsize, \pl)$ that
returns the union of all $\pl$-dominions of size at 
most~$\domsize+1$, including a winning strategy for $\pl$,
in time $O(\numprio \cdot m \cdot \binom{\domsize + \lceil \numprio/2 \rceil} {\domsize})$.
\end{lemma}

The lemma below summarizes some well-known facts about dominions and winning sets.
\begin{lemma}\label{lem:doms}
The following assertions hold for game graphs~$\game$ with
at least one outgoing edge per vertex 
and \upbr{in particular} parity objectives. Let $U \subseteq V$.
\begin{enumerate}
\item \cite[Lemma~4.4]{JurdzinskiPZ08}, \cite[Lemma~2.6.11]{Loitzenbauer16} 
Let $U$ be a $\pl$-trap in $\game$. Then 
a $\op$-dominion in $\game[U]$ is a $\op$-dominion in $\game$.\label{sublem:winclosed}
\item \cite[Lemma~4.1]{JurdzinskiPZ08} The set $W_\pl(\game)$ is a $\pl$-dominion.\label{sublem:winsetclosed}
\item \cite[Lemma~4.5]{JurdzinskiPZ08} 
Let $U$ be a subset of the winning set $W_\pl(\game)$ of player~$\pl$ and 
let $A$ be its $\pl$-attractor $\at{\pl}{\game}{U}$. Then the winning set 
$W_\pl(\game)$ of the player~$\pl$ is the union of $A$ and the winning set 
$W_\pl(\game[V \setminus A])$, % in the game graph induced by $V \setminus A$, 
and the winning set $W_\op(\game)$ of the opponent~$\op$ is equal to %his winning set 
$W_\op(\game[V \setminus A])$.% in the game graph induced by $V \setminus A$.
\label{sublem:subgraph}
\end{enumerate}
\end{lemma}

\subsection{Algorithm}
In this subsection we present our new algorithm to compute the winning sets of 
player~$\pe$ and player~$\po$ in a parity game $\pgame = (\game, \prio)$ 
with $\numprio \le \sqrt{n}$ priorities in time $O(n^{1+\gamma(\numprio+1)})$,
where $\gamma(\numprio) = \numprio/3 + 1/2 - 
4/(\numprio^2 - 1)$ for odd $\numprio$, and $\gamma(\numprio) = 
\numprio/3 + 1/2 - 1/(3 \numprio) - 4/\numprio^2$ for even $\numprio$. The slightly 
simpler special case for $\numprio = 3$ is described explicitly 
in the conference version~\cite{ChatterjeeHL15}.

\smallskip\noindent\emph{High level idea and relation to existing algorithms.}
The classical algorithm
for parity games \cite{McNaughton93,Zielonka98} identifies in each iteration 
a player-$\op$ dominion by a recursive call to a parity 
game with one priority less, then removes its $\op$-attractor from the game graph
and recurses on the remaining game graph until no $\op$-dominion can be found any more;
the remaining vertices form the winning set of player~$\pl$. The 
sub-exponential time algorithm of~\cite{JurdzinskiPZ08} limits the number of 
recursive calls by computing $\op$-dominions
with up to $\sqrt{n}$ vertices in a brute-force manner before each recursive call.
In the Big-Step algorithm \cite{Schewe07,Schewe17} this search for dominions before the recursive calls
is made more efficient by adapting the small progress measure algorithm for 
parity games~\cite{Jurdzinski00} to computing dominions of a bounded size.
Our algorithm refines the Big-Step algorithm of \cite{Schewe07,Schewe17} 
for dense graphs;
the main new insight is that to find dominions of size at most $\domsize$,
it is sufficient to consider a specific subgame with $O(n \cdot \domsize)$ edges.
For the special case of parity games with 2 priorities, i.e., 
Büchi games, our algorithm is equivalent to the algorithm of \cite{ChatterjeeH14}.
We next describe our algorithm and then prove its correctness and running time.
\begin{procedure}
\caption{Parity($\game$, $\prio$, $\domsize$)}
\label{proc:parity}
\SetKwInOut{Input}{Input}
\SetKwInOut{Output}{Output}
\SetKw{and}{and}
\BlankLine
\Input{\emph{game graph} $\game = ((V, E),(\ve, \vo))$ with $n = \lvert V \rvert$,\\
\emph{priority function} $\prio: V \rightarrow [\numprio]$ with $\numprio \ge 1$, and\\
\emph{parameter} $\domsize \in [1, n] \cap \mathbb{N}$}
\Output{winning sets $(\we, \wo)$ of player~$\pe$ and player~$\po$}
\BlankLine
\lIf{$\numprio = 1$}{\Return $(V, \emptyset)$}\label{l:startinit}
let $\pl$ be player~$\pe$ if $\numprio$ is odd and player~$\po$ otherwise\;
$W_{\op} \leftarrow \emptyset$\;\label{l:endinit}
\Repeat{$W'_{\op} = \emptyset$}{\label{l:startrepeat}
	$W'_{\op} \leftarrow \ref{proc:dom}(\game, \prio, 
	\domsize, \op)$\;
	\If{$W'_{\op} = \emptyset$ \and $\numprio \ge 3$}{
		$\game' \gets \game \setminus \at{\pl}{\game}{\prio^{-1}(\numprio-1)}$\;
		$(\we', \wo') \leftarrow \ref{proc:parity}(\game', \prio)$\;
	}
	$A \leftarrow \at{\op}{\game}{W'_{\op}}$\label{l:attrout}; 
	$W_{\op} \leftarrow W_{\op} \cup A$\;
	$\game \leftarrow \game \setminus A$\label{l:update}\;
}\label{l:endrepeat}
$W_{\pl} \gets V \setminus W_{\op}$\;
\Return{$(\we, \wo)$}
\end{procedure}

\smallskip\noindent\emph{Initialization \upbr{steps~\ref{l:startinit}--\ref{l:endinit} of Procedure~\ref{proc:parity}}.}
If all vertices of the parity game have priority zero, player~$\pe$ wins
from all vertices and thus the algorithm terminates and 
returns the set of all vertices as the winning set of player~$\pe$ and the empty 
set as the winning set of player~$\po$.
Otherwise we set $\pl$ according to the parity of the highest 
priority~$\numprio-1$ in the parity 
game. The procedure then iteratively determines the winning set of player~$\op$
and in the end identifies $W_\pl$ as the complement of $W_\op$.

\smallskip\noindent\emph{Iterated vertex deletions \upbr{steps~\ref{l:startrepeat}--\ref{l:endrepeat} of Procedure~\ref{proc:parity}}.}
The algorithm repeatedly removes vertices from the game graph $\game$. 
During the algorithm, we denote by $\game$ the remaining game graph after 
vertex deletions. The vertex removal is achieved by identifying parts of the 
winning set of player~$\op$, i.e, $\op$-dominions, and removing their attractors
from the maintained game graph.

\smallskip\noindent\emph{Dominion find and attractor removal.}
The algorithm repeatedly finds dominions of player~$\op$ in parity games~$\pgame'$
where the highest priority is at most $\numprio - 2$.
The parity game~$\pgame'$ is constructed by removing the $\pl$-attractor
of vertices with priority~$\numprio-1$ from $\game$ (this is implicit in 
the Procedure~\ref{proc:dom} and explicit before the recursive call
to Procedure~\ref{proc:parity}; more details follow).
After a $\op$-dominion in the parity game $\pgame'$ is found, its $\op$-attractor 
is removed from $\game$. Then the search for $\op$-dominions is continued on the 
remaining vertices. If all vertices in the parity game~$\pgame'$ are winning for $\pl$, 
i.e., no $\op$-dominion exists in~$\pgame'$, then the procedure terminates.
The winning set of player~$\op$ is the union of the $\op$-attractors of all 
found $\op$-dominions. The remaining vertices are winning for player~$\pl$.
We now describe the steps to find $\op$-dominions.

\begin{figure}
\centering
\begin{tikzpicture}
\matrix[column sep=8mm, row sep=8mm]{
	\node[even] (1) {$0$};
	& \node[odd] (2) {$0$};
	& \node[even,lightgray] (3) {$2$};
	& \node[odd,lightgray] (4) {$1$};
	& \node[odd,lightgray] (5) {$1$};\\
	\node[even] (6) {$1$};
	& \node[odd] (7) {$1$};
	& \node[odd] (8) {$0$};
	& \node[odd,lightgray] (9) {$2$};
	& \node[even,lightgray] (10) {$1$};\\
};
\path (1) edge[arrow] (2)
					edge[arrow] (7)
					edge[arrow] (8)
					edge[arrow, bend left] (6)
			(6) edge[arrow, bend left] (1)
			(2) edge[arrow, bend left] (7)
					edge[arrow,lightgray] (3)
			(7) edge[arrow] (6)
			(8) edge[arrow] (2)
			(9) edge[arrow,lightgray] (8)
					edge[arrow, bend left,lightgray] (3)
			(3) edge[arrow, bend left,lightgray] (9)
					edge[arrow, bend left,lightgray] (4)
					edge[arrow, bend left,lightgray] (5)
			(4) edge[arrow, bend left,lightgray] (3)
			(5) edge[arrow, bend left,lightgray] (10)
			(10) edge[arrow, bend left,lightgray] (5)
					edge[arrow,lightgray] (9);
\end{tikzpicture}
\caption{The resulting Büchi game $(\game', \prio)$ (in black) after removing 
$\at{\pe}{\game}{\prio^{-1}(2)}$ from~$\game$. 
The vertex in the top left corner is an
$\pe$-vertex with out-degree larger than $\sqrt{n}$, i.e., a blue vertex 
of $\game'_\idxg$ for $\idxg \le \log_2(\sqrt{n})$. The two Büchi vertices in the bottom left 
corner are contained in $\at{\pe}{\game'_\idxg}{\blue_\idxg}$ and we have 
$\mathcal{D}_\idxg = \emptyset$.}\label{fig:buchi}
\end{figure}

\begin{procedure}
\caption{Dominion($\game$, $\prio$, $\domsize$, $\op$)}
\label{proc:dom}
\SetKwInOut{Input}{Input}
\SetKwInOut{Output}{Output}
\BlankLine
\Input{\emph{game graph} $\game = ((V, E),(\ve, \vo))$, \\ 
\emph{priority function} $\prio: V \rightarrow[\numprio]$ with $\numprio \ge 2$, \\
\emph{parameter} $\domsize \in [1, n] \cap \mathbb{N}$, and \\
player~$\op$}
\Output{
a $\op$-dominion that contains all $\op$-dominions 
with at most $\domsize$ vertices or possibly the empty set 
if no such $\op$-dominion exists
}
\BlankLine
\For{$\idxg \leftarrow 1$ \KwTo $\lceil\log_2(\domsize)\rceil$}{
	construct $\game_\idxg$\; 
	$\blue_\idxg \leftarrow \set{v \in V_{\op} \mid \OutDeg(G_\idxg, v) = 0} \cup
	\{v \in V_{\pl} \mid \OutDeg(G, v) > 2^\idxg\}$\;
	$\game'_\idxg \gets \game_\idxg \setminus \at{\pl}{\game_\idxg}{\prio^{-1}(\numprio-1) \cup \blue_\idxg}$\;
	\If{$\numprio = 2$}{ 
		$D_\idxg \leftarrow$ the vertex set of $\game'_\idxg$\;e
	}
	\Else{
		$D_\idxg \leftarrow \progress(\game'_\idxg, \prio, 2^\idxg, \op)$\;
	}
	\If{$D_\idxg \ne \emptyset$}{
		\Return{$D_\idxg$}\;
	}
}
\Return{$\emptyset$}
% 	\;}
\end{procedure}

\smallskip\noindent\emph{Steps of dominion find.}
For the search for $\op$-dominions in the parity game~$\pgame'$ we use two 
different procedures, Procedure~\ref{proc:dom} and a recursive call to
Procedure~\ref{proc:parity}. 
We first search for ``small'' $\op$-dominions with up to 
$\domsize$ vertices with Procedure~\ref{proc:dom}, 
where $\domsize$ is a parameter that will be set later 
to balance the running times of the two procedures. 
If no $\op$-dominion is found, we can conclude that either 
all $\op$-dominions contain more than $\domsize$ vertices or 
the winning set of $\op$ on the current game graph is empty. In the latter case
the algorithm terminates. The former case occurs at most 
$n / \domsize$ times and in such a case we use the recursive
call to Procedure~\ref{proc:parity} on the parity game~$\pgame'$ with one priority less
to obtain a $\op$-dominion. Below we describe the details of Procedure~\ref{proc:dom}.

\begin{figure}
\centering
\begin{tikzpicture}
\matrix[column sep=8mm, row sep=8mm]{
	\node[even] (1) {$0$};
	& \node[odd] (2) {$0$};
	& \node[even,lightgray] (3) {$2$};
	& \node[odd,lightgray] (4) {$1$};
	& \node[odd,lightgray] (5) {$1$};\\
	\node[even] (6) {$1$};
	& \node[odd] (7) {$1$};
	& \node[odd] (8) {$0$};
	& \node[odd] (9) {$2$};
	& \node[even,lightgray] (10) {$1$};\\
};
\path (1) edge[arrow] (2)
					edge[arrow] (7)
					edge[arrow] (8)
					edge[arrow, bend left] (6)
			(6) edge[arrow, bend left] (1)
			(2) edge[arrow, bend left] (7)
					edge[arrow,lightgray] (3)
			(7) edge[arrow] (6)
			(8) edge[arrow] (2)
			(9) edge[arrow] (8)
					edge[arrow, bend left,lightgray] (3)
			(3) edge[arrow, bend left,lightgray] (9)
					edge[arrow, bend left,lightgray] (4)
					edge[arrow, bend left,lightgray] (5)
			(4) edge[arrow, bend left,lightgray] (3)
			(5) edge[arrow, bend left,lightgray] (10)
			(10) edge[arrow, bend left,lightgray] (5)
					edge[arrow,lightgray] (9);
\end{tikzpicture}
\caption{The winning set of player~$\po$ in the Büchi game 
$(\game', \prio)$ and
its $\po$-attractor in $\game$.}\label{fig:largedom}
\end{figure}

\begin{figure}
\centering
\begin{tikzpicture}
\matrix[column sep=8mm, row sep=8mm]{
	\node[even,lightgray] (1) {$0$};
	& \node[odd,lightgray] (2) {$0$};
	& \node[even] (3) {$2$};
	& \node[odd] (4) {$1$};
	& \node[odd] (5) {$1$};\\
	\node[even,lightgray] (6) {$1$};
	& \node[odd,lightgray] (7) {$1$};
	& \node[odd,lightgray] (8) {$0$};
	& \node[odd,lightgray] (9) {$2$};
	& \node[even] (10) {$1$};\\
};
\path (1) edge[arrow,lightgray] (2)
					edge[arrow,lightgray] (7)
					edge[arrow,lightgray] (8)
					edge[arrow, bend left,lightgray] (6)
			(6) edge[arrow, bend left,lightgray] (1)
			(2) edge[arrow, bend left,lightgray] (7)
					edge[arrow,lightgray] (3)
			(7) edge[arrow,lightgray] (6)
			(8) edge[arrow,lightgray] (2)
			(9) edge[arrow,lightgray] (8)
					edge[arrow, bend left,lightgray] (3)
			(3) edge[arrow, bend left,lightgray] (9)
					edge[arrow, bend left] (4)
					edge[arrow, bend left] (5)
			(4) edge[arrow, bend left] (3)
			(5) edge[arrow, bend left] (10)
			(10) edge[arrow, bend left] (5)
					edge[arrow,lightgray] (9);
\end{tikzpicture}
\caption{The remaining graph. The two vertices on the right are an $\po$-dominion, the remaining vertices form the winning set of player~$\pe$.}\label{fig:remain}
\end{figure}

\begin{example}[Illustration of the Algorithm for Parity-3.]
Figure~\ref{fig:game} shows a parity game with priorities $\set{0, 1, 2}$,
and Figure~\ref{fig:buchi} shows the Büchi game we obtain
when we remove the $\pe$-attractor of the vertices with the highest priority~$2$.
Figure~\ref{fig:largedom} shows the winning set of player~$\po$ in the Büchi game,
which is an $\po$-dominion in the parity game, 
and its $\po$-attractor in the original game graph.
Finally, Figure~\ref{fig:remain} shows the remaining game graph after
the removal of the $\po$-dominion and its attractor.
\end{example}

\smallskip\noindent\emph{Graph decomposition for Procedure~\ref{proc:dom}.}
In Procedure~\ref{proc:dom} we use the following \emph{hierarchical graph decomposition}.
For a game graph $\game = (G, (\ve, \vo))$ with $G = (V, E)$
we denote its decomposition with respect to player~$\op$ by $\set{\game_\idxg}$. 
The decomposition $\set{\game_\idxg}$ consists of $\lceil \log_2 n \rceil$ 
game graphs $\game_\idxg = (G_\idxg, (\ve, \vo))$ with underlying graphs
$G_\idxg = (V, E_\idxg)$. We call $1 \le \idxg \le \lceil \log_2 n \rceil$
the \emph{level} of $\game_\idxg$. For the definition of $E_\idxg$ we consider the 
incoming edges of each vertex in a fixed order: First the edges from vertices of
$V_\op$, then the remaining edges.
The set of edges $E_\idxg$ contains for each vertex $v \in V$ with 
$\OutDeg(G, v) \le 2^\idxg$ all its outgoing edges in $E$ and in addition for each
vertex $v \in V$ its first $2^\idxg$ incoming edges in $E$.
Note that (1)~$E_\idxg \subseteq E_{\idxg+1}$, 
(2)~$\lvert E_\idxg\rvert \le 2^{\idxg+1} n$, and 
(3)~$\game_{\lceil \log_2{n} \rceil} = \game$. 
We color $\pl$-vertices $v$ with $\OutDeg(G, v) > 2^\idxg$ and $\op$-vertices
without outgoing edges in~$\game_\idxg$ \emph{blue} 
in~$\game_\idxg$ and denote the set of blue vertices by $\blue_\idxg$. 
All other vertices are called \emph{white}.

\smallskip\noindent\emph{Procedure~\ref{proc:dom}.}
The Procedure~$\domalg$ searches 
in the game graph $\game_\idxg$, starting at $\idxg = 1$.
As long as no $\op$-dominion is found, the level~$\idxg$ is increased 
one by one up to at most $\idxg = \lceil\log_2(\domsize)\rceil$.
At each level only a subgame of $\game_\idxg$ that only contains white vertices 
is considered.
This subgame is obtained by removing the $\pl$-attractor of the  
blue vertices $\blue_\idxg$ and the vertices with the maximum priority
$\numprio-1$ from $\game_\idxg$.
This (a)~reduces reduce the number of priorities in the parity game by one
and (b)~ensures that $\op$-dominions found in the subgame are also 
$\op$-dominions in $\game$. In the subgame of $\game_\idxg$ for $\numprio \ge 3$
the $\progress$ algorithm (Lemma~\ref{lem:progress}) is used to find 
$\op$-dominions of size at most $O(2^\idxg)$. For $\numprio = 2$, i.e.,
Büchi games, we have that any non-empty set of remaining vertices is a dominion
of player~$\op$.

\smallskip\noindent\emph{Outline correctness.}
The correctness of Procedure~\ref{proc:parity} follows from the correctness
of the Big-Step algorithm as soon as we have established the correctness
of Procedure~\ref{proc:dom}. 
Let $V'_\idxg$ be the vertices of $\game'_\idxg = \game_\idxg \setminus \at{\pl}{\game_\idxg}{\prio^{-1}(\numprio-1) \cup \blue_\idxg}$.
The correctness of Procedure~\ref{proc:dom}
follows from the following arguments:
\begin{enumerate}
	\item the vertex set $V'_\idxg$ is a $\pl$-trap in $\game_\idxg$ and
	thus a $\op$-dominion of $\game'_\idxg$ is
	a $\op$-dominion of $\game_\idxg$;
	\item the vertices of $V'_\idxg$ are white and hence a $\op$-dominion
	of $\game_\idxg$ is a $\op$-dominion of $\game$; and
	\item thus the correctness of Procedure~\ref{proc:dom}
	follows from the correctness of the progress measure algorithm that is 
	used for determining dominions in $\game'_\idxg$.
\end{enumerate}

\smallskip\noindent\emph{Outline running time.}
We first analyze the running time without the recursive calls 
and then show a bound of $O(n^{1+\gamma(\numprio+1)})$
for the running time including the recursive calls by 
induction over~$\numprio$. For the latter we use that we only do 
recursive calls when no $\op$-dominion of size at most $\domsize$ exists 
and we balance the cost of the recursive calls and the other operations
by setting the parameter~$\domsize$ accordingly, similar to \cite{Schewe07}.
The running time without the recursive calls is, for $\numprio \ge 3$, dominated
by the running time of Procedure~\ref{proc:dom}.
For Procedure~\ref{proc:dom} we show that any $\op$-dominion that, including
its $\op$-attractor, contains at most $2^\idxg$ vertices can be identified 
as a $\op$-dominion in $\game_\idxg$. Thus if some $\op$-dominion is found 
in the search for $\op$-dominions at level $\idxg$ but was not found at 
level~$\idxg-1$, then it must contain $\Omega(2^\idxg)$ vertices. Since 
$\game_\idxg$ contains $O(2^\idxg \cdot n)$ edges, we can account for the 
time that was spent proportional to the edges of $\game_\idxg$ by charging it,
with an additional factor of $n$, to the vertices in the $\op$-dominion. This 
is crucial for the improved running time on dense graphs.

\bigskip
In the remaining part of this section we prove the correctness and running time
of Procedure~\ref{proc:parity} (including the calls to Procedure~\ref{proc:dom}).
The following lemma captures the essence of why the hierarchical graph 
decomposition is helpful for graph games. The lemma is a generalization of 
related lemmata for Büchi games in~\cite{ChatterjeeH14}
and for parity-3 games in the conference version of this paper~\cite{ChatterjeeHL15}
and was also used in subsequent work \cite{ChatterjeeDHL16b}.
The first part of the lemma is essential 
for the correctness of the hierarchical graph decomposition technique on 
game graphs: It shows that every $\pl$-trap in $\game_\idxg$ that contains
only white vertices is also a $\pl$-trap in $\game$.
The second part is essential for the running time: Every $\pl$-trap
in $\game$ that, including its $\op$-attractor, contains at most $\domsize$
vertices, is a $\pl$-trap induced by white vertices in 
$\game_{\lceil \log_2 \domsize \rceil}$. Note that the $\op$-attractor of a
$\pl$-trap is itself a $\pl$-trap and thus this holds in particular
for maximal $\pl$-traps of size at most $\domsize$.

\begin{lemma}
\label{lem:p:decomp}
	Let $\game = ((V, E), (\ve, \vo))$ be a game graph and $\set{\game_\idxg}$ its 
	hierarchical graph decomposition w.r.t.\ to player~$\op$. 
	For $1 \le i \le \lceil \log_2 n \rceil$ let 
	$\blue_\idxg$ be the set consisting of the player-$\op$ vertices that have no outgoing edge
	in $\game_\idxg$ and the player-$\pl$ vertices with more than $2^\idxg$ outgoing edges in~$\game$. 
	\begin{enumerate}%[\textup{(}1\textup{)}]
		\item If a set $S \subseteq V \setminus \blue_\idxg$ 
		is a $\pl$-trap in $\game_\idxg$, then $S$ is a $\pl$-trap in~$\game$.
		\label{sublem:p:sound}
		\item If a set $S \subseteq V$ is 
		a $\pl$-trap in $\game$ and $\lvert \at{\op}{\game}{S} \rvert \le 2^\idxg$,
		then \upbr{i} $S$ is a $\pl$-trap in~$\game_\idxg$, \upbr{ii} the set $S$ is in $V \setminus \blue_\idxg$, and
	  \upbr{iii} $\game_\idxg[S] = \game[S]$.
		\label{sublem:p:size}
	\end{enumerate}
\end{lemma}
\begin{proof}\hfill
	\begin{enumerate}
		\item By $S \subseteq V \setminus \blue_\idxg$ we have for all $v \in S \cap V_\pl$
		that $\Out(G, v) = \Out(G_\idxg, v)$.
		Thus if $\Out(\game_\idxg, v) \subseteq S$, then also $\Out(G, v) \subseteq S$.
		Each edge of $G_\idxg$ is contained in $G$, thus we have 
		for all $v \in S \cap V_\op$ that $\Out(G_\idxg, v) \cap S \ne \emptyset$ 
		implies $\Out(G, v) \cap S \ne \emptyset$.
		
		\item By the definition of a $\op$-attractor, all $u \in V_\op$ for which
		there exists a $w \in S$ such that $(u, w) \in E$ are contained in $\at{\op}
		{\game}{S}$. Thus by $\lvert \at{\op}{\game}{S} \rvert \le 2^\idxg$
		all vertices of $S$ have less than $2^\idxg$ incoming edges from vertices 
		of $\op$. Hence by the ordering of the incoming edges in the construction
		of $\game_\idxg$ 
		\begin{enumerate}
			\item[(a)] all edges $(u, w)$ with $u \in V_\op$ and $w \in S$
			are contained in $E_\idxg$ and thus in particular
			\item[(b)] $(u, w) \in E_\idxg$ for all $(u, w) \in E$ 
			with $u \in S \cap V_\op$ and $w \in S$.
		\end{enumerate}
		Note furthermore that since $S$ is a $\pl$-trap in $\game$, there 
		exists a $w \in S$ such that $(u, w) \in E$ for all $u \in S \cap V_\op$;
		by (a) we have $(u, w) \in E_i$ for this $w$.
		Together with $E_\idxg \subseteq E$ it follows that
		(i) $S$ is a $\pl$-trap in $\game_\idxg$.
		The existence of such an edge $(u, w)$ 
		also ensures that every vertex in $V_\op \cap S$ has an outgoing
		edge in $\game_\idxg$, i.e.,
		\begin{enumerate}
			\item[(c)] $\blue_\idxg \cap V_\op \cap S = \emptyset$. 
		\end{enumerate}
		
		Since $S$ is a trap for player~$\pl$ in $\game$ 
		and $\lvert S \rvert \le 2^\idxg$,
		we have for all $v \in V_\pl$ that $\OutDeg(G, v) < 2^\idxg$.
		Thus
		\begin{enumerate}
			\item[(d)] $\blue_\idxg \cap V_\pl \cap S = \emptyset$ and
			\item[(e)] $(u, w) \in E_\idxg$ for all $(u, w) \in E$ 
			with $u \in S \cap V_\pl$ and $w \in S$.
		\end{enumerate}
		Combining (c) and (d) yields (ii) $S \cap \blue_\idxg = \emptyset$,
		and (b) and (e) give (iii) $\game_\idxg[S] = \game[S]$.\qedhere
	\end{enumerate}
\end{proof}

\noindent In the next lemma we consider not just $\pl$-traps but $\op$-dominions 
and thus argue additionally about winning strategies of player~$\op$. The first 
part of the lemma shows the soundness of Procedure~\ref{proc:dom}, while
the second part shows completeness and is crucial for the running time analysis
of the overall algorithm.

\begin{lemma}\label{lem:p:domdecomp}
		Let the parity game $(\game,\prio)$ with $\numprio$ priorities, 
		the parameter $\domsize$, and the player~$\op$ be the input to
	Procedure~\ref{proc:dom}. 
	Let $D_\idxg$, $\game_\idxg'$ and $\blue_\idxg$ be as in Procedure~\ref{proc:dom}.
	\begin{enumerate}%[\textup{(}1\textup{)}]
		\item Every $D_\idxg \ne \emptyset$ is a $\op$-dominion in the parity 
		game $(\game,\prio)$
		and the priority of the vertices of $D_\idxg$ is at most $\numprio - 2$.
		\label{sublem:p:domsound}
		\item If there exists a $\op$-dominion $D$ with $\lvert 
		\at{\op}{\game}{D} \rvert \le 2^\idxg$ in $(\game, \prio)$ such that 
		the highest priority in $D$ is at most $\numprio - 2$,
		then $D$ is a $\op$-dominion in the parity game $(\game_\idxg', \prio)$.
		\label{sublem:p:domsize}
	\end{enumerate}
\end{lemma}

\begin{proof}\hfill
	\begin{enumerate}
		\item By definition the highest priority in $\game_\idxg'$ and thus in 
		$D_\idxg$ is at most $\numprio - 2$.
		Let $V_\idxg'$ be the vertices of $\game_\idxg'$. By 
		Lemma~\ref{lem:attr}~\eqref{sublem:complattr} the set $V_\idxg'$ is a
		$\pl$-trap in $\game_\idxg$ and by $V_\idxg' \cap \blue_\idxg = \emptyset$ and
		Lemma~\ref{lem:p:decomp}~\eqref{sublem:p:sound} also in $\game$.
		If $\numprio = 2$, we have $V_\idxg' = D_\idxg$ and 
		since $D_\idxg$ only contains vertices with priority~$0$,
		the set $D_\idxg$ is a $\op$-dominion in $(\game, \prio)$. 
		Suppose now $\numprio \ge 3$. Then $D_\idxg$ is the set returned by 
		$\progress(\game'_\idxg, \prio, 2^\idxg, \op)$ and we have by
		Lemma~\ref{lem:progress} that $D_\idxg$ is a $\op$-dominion
		in the parity game $(\game_\idxg', \prio)$. 
		Since $V_\idxg'$ is a $\pl$-trap in $\game_\idxg$,
		the set $D_\idxg$ is also a $\op$-dominion in $\game_\idxg$ by 
		Lemma~\ref{lem:doms}~\eqref{sublem:winclosed}. 
		Since $D_\idxg$ contains only white vertices,
		we have 
		\begin{enumerate}[label=\({\alph*}]
			\item the set $D_\idxg$ is a $\pl$-trap in $\game$ by 
		Lemma~\ref{lem:p:decomp}~\eqref{sublem:p:sound} and
			\item  all outgoing
		edges of vertices of $V_{\pl} \cap D_\idxg$ are present in $\game_\idxg$. 
		\end{enumerate}
		Thus by $E_\idxg \subseteq E$
		the winning strategy of player~$\op$ for the vertices of $D_\idxg$ in 
		$\game_\idxg$ is also a winning strategy in $\game$ and hence $D_\idxg$ 
		is a $\op$-dominion in the parity game $(\game, \prio)$.
		
		\item By Lemma~\ref{lem:p:decomp}~\eqref{sublem:p:size} we have \upbr{i} 
		$D$ is a $\pl$-trap in $\game_\idxg$, \upbr{ii} $D \cap \blue_\idxg = \emptyset$, and \upbr{iii} $\game[D] = \game_\idxg[D]$.
		Thus
		\begin{enumerate}[label=\({\alph*}]
			\item $D$ is contained in~$\game_\idxg'$ and 
			\item player~$\op$ can play the same winning strategy 
		in~$\game_\idxg[D]$ as in~$\game[D]$.\qedhere
		\end{enumerate}
	\end{enumerate}
\end{proof}

\noindent In the following corollary we state the insights of the previous two lemmata 
as needed for the running time analysis. The first part shows that when we use
the hierarchical graph decomposition with increasing level~$\idxg$
to search for a $\op$-dominion and we have to go up to level~$\idxg^*$, then
the found $\op$-dominion, or at least its $\op$-attractor (which is again a $\op$-dominion),
contains a number of vertices proportional to $2^{i^*}$, which allows us to charge 
the work done in the search to the vertices in the identified dominion. The second
part of the corollary states that no ``small'' $\op$-dominions
exist in the maintained parity game
if Procedure~\ref{proc:dom} returns the empty set, where ``small'' is 
specified by the parameter~$\domsize$ that will be set to balance
the running time of Procedure~\ref{proc:dom} and the recursive calls. 
In this case either no $\op$-dominion exists in the parity game and the 
algorithm terminates or the subsequent recursive call identifies a $\op$-dominion
with more than~$\domsize$ vertices; the latter can happen at most
$O(n / \domsize)$ times and allows us to bound the number of 
iterations in Procedure~\ref{proc:parity}.

Recall that we set $\pl$ to $\pe$ if the highest priority in the parity game is even
and to odd otherwise, i.e., when we search for $\op$-dominions, we search for
dominions of the player that tries to avoid the highest priority.
For the proof of the second part we use that in this case every 
$\op$-dominion~$D$ contains a subset $D'$ that is a $\op$-dominion
itself and does not contain any vertex with the highest priority $\numprio - 1$
(Proposition~2 of~\cite{Jurdzinski00}).
Intuitively, $D'$ is the set of vertices that are contained in
the cycles of $D$ that are induced by the memoryless winning strategy of 
player~$\op$ in $\game[D]$.

\begin{corollary}\label{cor:minsize}
		Let the parity game $(\game,\prio)$ with $\numprio$ priorities, 
		the parameter $\domsize$, and the player~$\op$ be the input to
	Procedure~\ref{proc:dom}. 
	Let $D_\idxg$, $\game_\idxg'$ and $\blue_\idxg$ be as in Procedure~\ref{proc:dom}.
	\begin{enumerate}%[\textup{(}1\textup{)}]
		\item If for some $\idxg > 1$ we have $D_\idxg \ne \emptyset$ but $D_{\idxg-1} = \emptyset$,
		then $\lvert \at{\op}{\game}{D_\idxg} \rvert > 2^{\idxg-1}$.
		\item If Procedure~\ref{proc:dom} returns the empty set, then we have 
		for every $\op$-dominion~$D$ in the given parity game $\lvert 
		\at{\op}{\game}{D} \rvert > \domsize$.
	\end{enumerate}
\end{corollary}

\begin{proof}\noindent
	\begin{enumerate}
		\item By Lemma~\ref{lem:p:domdecomp}~\eqref{sublem:p:domsound} we have
		that $D_\idxg$ is a $\op$-dominion in $(\game, \prio)$ and the vertices
		of $D_\idxg$ have a priority of at most $\numprio - 2$.
		Assume by contradiction that $\lvert \at{\op}{\game}{D_\idxg} \rvert \le 2^{\idxg-1}$.
		Then by Lemma~\ref{lem:p:domdecomp}~\eqref{sublem:p:domsize}
		the set $D_{\idxg-1}$ contains $D_\idxg$ and thus 
		is not empty, a contradiction.
	\item Assume by contradiction that Procedure~\ref{proc:dom} returns the empty
	set and there exists a $\op$-dominion~$D$ with
	$\lvert \at{\op}{\game}{D} \rvert \le \domsize$ in 
	$(\game, \prio)$. By Proposition~2 of~\cite{Jurdzinski00} 
	in this case also a $\op$-dominion $D' \subseteq D$ exists 
	such that the vertices of $D'$ have priority at most $\numprio - 2$.
  By Lemma~\ref{lem:p:domdecomp}~\eqref{sublem:p:domsize}
	the set $D'$ is a $\op$-dominion in the parity game $(\game_\idxg', \prio)$
	for $\idxg \ge \lceil\log_2(\lvert D' \rvert) \rceil$ and thus in particular for 
	$\idxg = \lceil\log_2(\domsize)\rceil$. Hence
	by Lemma~\ref{lem:progress} Procedure~\ref{proc:dom} returns 
	a non-empty set, a contradiction.
	\qedhere
	\end{enumerate}
\end{proof}

\noindent In the following lemma we use Lemma~\ref{lem:progress} that bounds the running time 
of the calls to the progress measure algorithm together with 
the relation between levels in the hierarchical
graph decomposition and the size of $\op$-dominions, provided by the corollary
above, to bound the running time of Procedure~\ref{proc:parity} without recursive calls.

\begin{lemma}\label{lem:timedom}
Let $(\game, \prio)$ be a parity game with a game graph $\game = ((V, E),(\ve, \vo))$ 
with $n = \lvert V \rvert$, a priority function $\prio$, and $\numprio$ priorities
and let $\domsize \in [1, n]$ be a parameter that for $\numprio = 2$ is equal 
to $\domsize = n$.
The running time of Procedure~\ref{proc:parity}\upbr{$\game$, $\prio$} \emph{without the 
recursive calls}, and without the attractor computation before the recursive calls,
is $O\big(\numprio  \cdot n^2 \cdot 
\binom{\domsize + \lfloor \numprio/2 \rfloor}{\domsize}\big)$
for $\numprio \ge 3$, $O(n^2)$ for 
$\numprio = 2$, and $O(n)$ for $\numprio = 1$.
\end{lemma}
\begin{proof}
	All operations before and after the repeat-until loop can be done in $O(n)$ time,
	which shows $O(n)$ for $\numprio = 1$.
	Further the attractor computations and the updates of the maintained sets in 
	lines~\ref{l:attrout}--\ref{l:update} can be done in total time $O(m)$.
	Thus it remains to bound the total time for the calls to Procedure~\ref{proc:dom}.
	
	To efficiently construct the graphs $\game_\idxg$ and the vertex sets $\blue_\idxg$, we 
	maintain ordered lists of the incoming and outgoing edges of each vertex. These lists
	can be updated whenever an obsolete entry is encountered in the construction
	of $\game_\idxg$; as each entry is removed at most once, this takes total time $O(m)$.
	
	We now analyze the time spent in an iteration~$\idxg$ of the for-loop in 
	Procedure~\ref{proc:dom}. The graph $\game_\idxg$ contains $O(2^\idxg \cdot n)$
	edges and both $\game_\idxg$ and $\blue_\idxg$ can be constructed from the maintained 
	lists of in- and outedges in $O(2^\idxg \cdot n)$ time. Also the attractor 
	computation takes time $O(2^\idxg \cdot n)$. Thus for $\numprio = 2$ the time
	in iteration~$\idxg$ is $O(2^\idxg \cdot n)$, while for $\numprio \ge 3$ the time
	is dominated by the call to $\progress$. Note that with the attractor computation 
	the vertices with the highest priority are removed from the parity game, thus 
	the call to $\progress$ is done for a parity game with $\numprio - 1$ priorities
	and parameter $\domsize = 2^\idxg$. Hence by Lemma~\ref{lem:progress} 
	iteration~$\idxg$ for $\numprio \ge 3$ takes time 
	\begin{equation*}O\left(
	\numprio \cdot n \cdot 2^\idxg \cdot
	\binom{2^\idxg + \lceil (\numprio - 1)/2 \rceil}{2^\idxg}
	\right) = O\left(
	\numprio \cdot n \cdot 2^\idxg \cdot
	\binom{2^\idxg + \lfloor \numprio/2 \rfloor}{2^\idxg}
	\right)\,.\end{equation*}
	The time for all iterations up to the $\idxg$-th iteration forms a geometric series 
	and thus satisfies the same running time bound. 
	
	Let $\idxg^*$ be the last iteration of the for-loop in a call to 
	Procedure~\ref{proc:dom}. Let $\pl$ be $\pe$ if $\numprio$ is odd and 
	$\po$ otherwise. By Corollary~\ref{cor:minsize} either 
	\begin{enumerate}%[(1)]
		\item  $D_{i^*}$ is a 
	$\op$-dominion with $\lvert \at{\op}{\game}{D_{i^*}} \rvert > 2^{i^* - 1}$ 
	vertices or
		\item $\idxg^* = \lceil \log_2(\domsize )\rceil$ and
	$\game$ does not contain any $\op$-dominion~$D$ with $\at{\op}{\game}{D} 
	 \le \domsize$ vertices.	
	\end{enumerate} 
	In case~\upbr{1} more than $2^{i^* - 1}$ %< \domsize
	vertices are removed from the maintained graph in this iteration. We charge
	each of these vertices 
	$O\big(
	\numprio \cdot n \cdot
	\binom{2^\idxg + \lfloor \numprio/2 \rfloor}{2^\idxg}
	\big)$
	time, which can be bounded by 
	$O\big(
	\numprio \cdot n \cdot
	\binom{\domsize + \lfloor \numprio/2 \rfloor}{\domsize}
	\big)$
	(per vertex). Hence the total running time is bounded by
	\begin{equation*}O\left(
	\numprio \cdot n^2 \cdot 
	\binom{\domsize + \lfloor \numprio/2 \rfloor}{\domsize}
	\right)\,.
	\end{equation*}
	In case~\upbr{2} either 
	\begin{enumerate}[label=\({\alph*}]
		\item $\numprio \ge 3$ and a $\op$-dominion
	with more than $\domsize$ vertices in its $\op$-attractor is detected
	in the subsequent recursive call to Procedure~\ref{proc:parity} or
		\item there is no $\op$-dominion in the maintained parity game
	and this is the last iteration of the repeat-until loop in the 
	Procedure~\ref{proc:parity}.
	\end{enumerate}
	In Case~\upbr{a} we have again that more than $2^{i^* - 1}$ %< \domsize
	vertices are removed from the maintained graph in this iteration and thus
	can apply the same argument as for Case~\upbr{1}.
	Case~\upbr{b} can happen at most once and its running time is bounded by $O(n^2)$
	for $\numprio = 2$ and by 
	\begin{equation*}O\left(
	\numprio \cdot n \cdot 2^{\log_2(\domsize)} \cdot 
	\binom{2^{\log_2(\domsize)} + \lfloor \numprio/2 \rfloor}{2^{\log_2(\domsize)}}
	\right)\end{equation*}
	 for $\numprio \ge 3$,
	which can be bounded by
% 	\begin{equation*}O\left(
	$O\big(
	\numprio \cdot n^2 \cdot 
	\binom{\domsize  + \lfloor \numprio/2 \rfloor}{\domsize}
	\big)$.
\end{proof}

To bound the running time including the recursive calls, we use a similar analysis
and similar parameters as 
for the Big-Step algorithm in~\cite{Schewe07,Schewe17}.
Let $\gamma(\numprio) = \numprio/3 + 1/2 - 
4/(\numprio^2 - 1)$ for odd $\numprio$ and $\gamma(\numprio) = 
\numprio/3 + 1/2 - 1/(3 \numprio) - 4/\numprio^2$ for even $\numprio$.
Further let $\beta(\numprio) = \gamma(\numprio)/(\lfloor\numprio / 2 \rfloor + 1)$.
It can easily be verified that
$\gamma(\numprio+1) = 1 + \gamma(\numprio)- \beta(\numprio)$ holds.
We set $\domsize = n^{\beta(\numprio)}$ for $\numprio \ge 3$ and
additionally $\domsize = n$ for $\numprio = 2$. We show by induction over $\numprio$
a running time bound of $O(n^{1+\gamma(\numprio+1)})
= O(n^{2 + \gamma(\numprio) - \beta(\numprio)})$ for parity games with 
$\numprio$~priorities. The running time of the Big-Step algorithm for parity 
games with $\numprio$ priorities is
$O(m \cdot (6 e^{5/3} n / c^2)^{\gamma(\numprio)})$, 
i.e., for a constant number of priorities we replace $m$ by $n^{2-\beta(\numprio)}$.
\begin{lemma}[Running time]\label{lem:time}
	For parity games with $\numprio \le \sqrt{n}$ priorities
	Procedure~\ref{proc:parity} takes time 
	$O(n^{1+\gamma(\numprio+1)})$, where 
	$\gamma(\numprio) = \numprio/3 + 1/2 - 4/(\numprio^2 - 1)$ for odd~$\numprio$ 
	and $\gamma(\numprio) = \numprio/3 + 1/2 - 1/(3 \numprio) - 4/\numprio^2$ for 
	even~$\numprio$.
\end{lemma}
\begin{proof}
For the base case of $\numprio = 2$ we have $\gamma(\numprio + 1) = 1$ and no recursive 
calls. Thus the running time of Procedure~\ref{proc:parity} for $\numprio = 2$ 
is $O(n^2)$ by Lemma~\ref{lem:timedom} (in this case Procedure~\ref{proc:parity}
is equivalent to the algorithm of~\cite{ChatterjeeH14}). Suppose Procedure~\ref{proc:parity}
runs in time $O(n^{1+\gamma(\numprio)})$ for a parity game with 
$\numprio - 1 \ge 2$ priorities. We show that this implies that 
Procedure~\ref{proc:parity} runs in time $O(n^{1+\gamma(\numprio+1)})$
for a parity game with $\numprio$ priorities for $3 \le \numprio \le \sqrt{n}$. Let 
$\domsize = n^{\beta(\numprio)}$ for $\beta(\numprio) 
= \gamma(\numprio)/(\lfloor\numprio / 2 \rfloor + 1)$.
We have $\beta(\numprio) \ge 1/2$ for all $\numprio \ge 3$ and 
thus $\domsize \ge \sqrt{n}$. By Lemma~\ref{lem:timedom}
the time spent in Procedure~\ref{proc:parity}
without the recursive calls is $O\big(\numprio  \cdot n^2 \cdot 
\binom{\domsize + \lfloor \numprio/2 \rfloor}{\domsize}\big)$. With 
Stirling's approximation of $(x/e)^x \le x!$ we have
\begin{align*}
\binom{\domsize + \lfloor \numprio/2 \rfloor}{\domsize} 
&\le  \frac{\left(\domsize + \lfloor \numprio/2 \rfloor \right)^{\lfloor \numprio/2 \rfloor}}{\lfloor \numprio/2 \rfloor !}\,,\\ 
&\le \left(\frac{\left(\domsize + \lfloor \numprio/2 \rfloor \right) \cdot e}{\lfloor \numprio/2 \rfloor}\right)^{\lfloor \numprio/2 \rfloor}\,.
\end{align*}
Using $3 \le \numprio \le \sqrt{n} \le \domsize$, we obtain
\begin{align*}
 \left(\frac{\left(\domsize + \lfloor \numprio/2 \rfloor \right) \cdot e}{\lfloor \numprio/2 \rfloor}\right)^{\lfloor \numprio/2 \rfloor}
&\le \left(\frac{2 e \domsize + e \numprio}{\numprio - 1}\right)^{\lfloor \numprio/2 \rfloor}\,,\\
&\le \left(\frac{5 e \domsize}{\numprio}\right)^{\lfloor \numprio/2 \rfloor}\,.\\
\end{align*}
Thus we have
\begin{align*}
\numprio \cdot \binom{\domsize + \lfloor \numprio/2 \rfloor}{\domsize}
&\le  \frac{(5 e)^{\lfloor \numprio/2 \rfloor}}{\numprio^{\lfloor \numprio/2 \rfloor - 1}} \domsize^{\lfloor \numprio/2 \rfloor}\,,\\
\end{align*}
which is $\le \domsize^{\lfloor \numprio/2 \rfloor}$ 
for $\numprio \ge (5 e)^{3/2}$ and 
 $\le \kappa \domsize^{\lfloor \numprio/2 \rfloor}$ for some constant $\kappa$
for $\numprio < (5 e)^{3/2}$.
Hence the time without the recursive calls is bounded by $O(n^2 \cdot 
\domsize^{\lfloor \numprio/2 \rfloor})$, which is equal to $
O(n^{2 + \beta(\numprio)\lfloor \numprio/2 \rfloor})$ for $\domsize = n^{\beta(\numprio)}$.
By the choice of $\beta(\numprio)$ we have $\beta(\numprio)\lfloor \numprio/2 \rfloor
= \gamma(\numprio) - \beta(\numprio)$ and thus we can write this bound
as $O(n^{2 + \gamma(\numprio) - \beta(\numprio)})$.
By Corollary~\ref{cor:minsize} there are at most $O(n / \domsize)
= O(n^{1 - \beta(\numprio)})$ recursive calls to Procedure~\ref{proc:parity}. 
Each recursive call is for a parity game 
with one priority less and thus takes time $O(n^{1+\gamma(\numprio)})$.
Hence the total time for all recursive calls is bounded by 
$O(n^{1 - \beta(\numprio) + 1 + \gamma(\numprio)})$. 
For $\gamma(\numprio)$ as defined above we have $\gamma(\numprio+1) = 1 +
\gamma(\numprio)- \beta(\numprio)$, which completes the proof. %have verified that (V)
\end{proof}

The correctness proof for Procedure~\ref{proc:parity} follows mostly the 
correctness proof of the classical algorithm for parity games and is folklore; 
additionally Lemma~\ref{lem:p:domdecomp}~\eqref{sublem:p:domsound} shows the 
correctness of Procedure~\ref{proc:dom}.

\begin{lemma}[Correctness]\label{lem:correctness}
Given a parity game $\pgame = (\game,\prio)$, let $\pl$ be player~$\pe$ if
	$\numprio$ is odd and player~$\po$ otherwise and let 
	$W_\pl$ and $W_\op$ be the output of Procedure~\ref{proc:parity}.
We have: \upbr{1}~\emph{(Soundness).} $W_\op \subseteq W_\op(P)$; and 
\upbr{2}~\emph{(Completeness).} $W_\op(P) \subseteq W_\op$.
\end{lemma}
\begin{proof}
The proof is by induction over the number of priorities~$\numprio$.
The induction base is $\numprio = 1$, where the algorithm correctly returns
$\we = V$ and $\wo = \emptyset$. Suppose the algorithm is correct for parity
games with $\numprio - 1$ priorities. We show the correctness for
parity games with $\numprio$ priorities by proving \upbr{1} 
by induction on the iterations 
of the repeat-until loop that all vertices of $W_\op$ are indeed winning for 
player~$\op$ (soundness) and then \upbr{2} 
construct a winning strategy for player~$\pl$ on the 
remaining vertices $W_\pl = V \setminus W_\op$ (completeness).

Soundness will follow from showing that whenever the set $W'_\op$ determined
in an iteration of the repeat-until loop is not empty, then it is a $\op$-dominion.
This is sufficient because by 
Lemma~\ref{lem:doms}~\eqref{sublem:subgraph}
it is valid to determine a $\op$-dominion, remove its $\op$-attractor,
and recurse on the remaining game graph; and hence  $W'_\op$  being a $\op$-dominion
implies soundness by induction over the repeat-until loop.
If $W'_\op \ne \emptyset$ is returned by Procedure~\ref{proc:dom},
then it is a $\op$-dominion by Lemma~\ref{lem:p:domdecomp}~\eqref{sublem:p:domsound}.
If $W'_\op \ne \emptyset$ is returned by the recursive call to 
Procedure~\ref{proc:parity} for the game graph $\game' = \game \setminus 
\at{\pl}{\game}{\prio^{-1}(\numprio-1)}$, then $W'_\op$ is a $\op$-dominion
by the following argument: By the induction assumption the set $W'_\op$ is 
a $\op$-dominion in $\game'$. By Lemma~\ref{lem:attr}~\eqref{sublem:complattr}
the vertices of $\game'$ form a $\pl$-trap in $\game$ and by 
Lemma~\ref{lem:doms}~\eqref{sublem:winclosed} a $\op$-dominion in a subgame 
induced by a $\pl$-trap is a $\op$-dominion in the full game.

We now prove the completeness result.
When Procedure~\ref{proc:parity} terminates, the winning set~$W'_\op$ of player~$\op$ 
in the parity game $(\game',\prio)$ is empty. 
Also note that since the algorithm removes attractors of $\op$, the set $W_\pl$ 
is a trap for $\op$ by Lemma~\ref{lem:attr}~(\ref{sublem:complattr}).
Consider the set $Z = \set{v \in W_\pl \mid \prio(v) = \numprio-1}$, its 
attractor $X =\at{\pl}{\game}{Z}$, and the subgame induced by 
$U = W_\pl \setminus X$. 
Note that the game graphs $\game[U]$ and $\game'[U]$ coincide. 
Thus all vertices of $U$ must be winning for player~$\pl$ in the parity game 
$(\game',\prio)$ as otherwise $W'_\op$ would have been non-empty for 
$(\game',\prio)$.
We prove the lemma by describing a winning strategy for player~$\pl$ in $\pgame$ 
for all vertices in $W_\pl$. For vertices of $Z \cap V_{\pl}$ the winning 
strategy chooses an edge in $W_\pl$, which exists since $W_\pl$ is a $\op$-trap.
For vertices in $X\setminus Z$ player~$\pl$ follows her attractor strategy to $Z$. 
In the subgame induced by $U = W_\pl \setminus X$ player~$\pl$ follows her
winning strategy in the parity game $(\game',\prio)$. Then in a play 
either (i)~$X$ is visited infinitely often; 
or (ii)~from some point on only vertices of $U$ are visited.
In the former case, the attractor strategy ensures that then some vertex of $Z$ 
with priority~$\numprio-1$ is visited infinitely often; and in the later case, 
the subgame winning strategy ensures that the highest priority visited 
infinitely often has the same parity as $\numprio-1$.
It follows that $W_\pl \subseteq W_\pl(P)$, i.e., $W_\op(P) \subseteq W_\op$, 
and the desired result follows.
\end{proof}

Lemmata~\ref{lem:time} and~\ref{lem:correctness} yield the following result.

\begin{theorem}\label{th:paritydense}
Procedure~\ref{proc:parity} correctly computes the winning sets in parity games
with $n$ vertices and
$\numprio \le \sqrt{n}$~priorities in $O(n^{1+\gamma(\numprio+1)})$ time,
where $\gamma(\numprio) = \numprio/3 + 1/2 - 4/(\numprio^2 - 1)$ for odd~$\numprio$ 
	and $\gamma(\numprio) = \numprio/3 + 1/2 - 1/(3 \numprio) - 4/\numprio^2$ for 
	even~$\numprio$.
\end{theorem}

\smallskip\noindent\emph{Computation of winning strategies.} 
In parity-3 games the previous results for computing winning strategies for the 
players in their respective winning sets are as follows:
The small-progress measure algorithm of~\cite{Jurdzinski00} requires $O(m n)$~time 
to compute the winning strategy of player~$\pe$ and $O(m n^2)$~time to 
compute the winning strategy for player~$\po$; Schewe~\cite{Schewe17} 
shows how to modify the small-progress measure algorithm to compute the 
respective winning strategies of both players in 
$O(m n)$~time. Schewe's running time bound for general parity games also 
holds when both winning strategies are requested~\cite{Schewe17}.
We show that our algorithm also computes the respective 
winning strategies without increasing the running time, i.e., in 
$O(n^{5/2})$~time for parity-3 games and in $O(n^{1+\gamma(\numprio+1)})$~time 
for parity games with $\numprio \le \sqrt{n}$. 

For \emph{parity-3} we first observe that for Büchi games~\cite{ChatterjeeH14}
we can construct in $O(n^2)$~time also the respective winning strategies of both 
players since the algorithm is based on identifying traps and attractors, 
and the corresponding winning strategies are identified immediately with the 
computation.
The proof of Lemma~\ref{lem:correctness} describes the strategy computation for 
a winning strategy of player~$\po$ which involves an attractor strategy and the sub-game 
strategy for Büchi games, each of which can be computed in $O(n^2)$ time.
A winning strategy for player~$\pe$ is obtained in the iterations of the algorithm, 
i.e., whenever we obtain a dominion by solving Büchi games, we also obtain a 
corresponding winning strategy, and similarly for the attractor computation.
Thus the winning strategy for player~$\pe$ can be computed in $O(n^{5/2})$~time.
For \emph{general parity games} the winning strategies for both players are constructed
in a similar way; the argument uses parity-3 games as base case and then 
induction over the recursive calls.
Let $\pl$ be player~$\pe$ if $\numprio$ is odd and 
player~$\po$ otherwise. First note that the time bound in 
Lemma~\ref{lem:progress}, and therefore the time bound of Procedure~\ref{proc:dom},
includes the computation of a winning strategy for player~$\op$ within a 
$\op$-dominion determined by Procedure~\ref{proc:dom}.
 The winning strategy of player~$\op$ is a combination
of his winning strategies for the dominions identified in 
Procedure~\ref{proc:dom} and the dominions identified in the recursive 
calls for parity games with one priority less and the corresponding 
attractor strategies. The winning strategy of player~$\pl$, as described in 
Lemma~\ref{lem:correctness}, is identified in the last
iteration of the repeat-until loop and consists of her winning strategy 
for the parity game for which the last recursive call is made and 
her attractor strategy to vertices with priority~$\numprio-1$.

\begin{corollary}
Winning strategies for player~$\pe$ and player~$\po$ in their 
respective winning sets in parity games with $n$ vertices and
$\numprio \le \sqrt{n}$~priorities can be computed in $O(n^{1+\gamma(\numprio+1)})$ time. 
\end{corollary}

\section{Streett Objectives in Graphs}\label{sec:streett}

In this section we present our algorithm for graphs with Streett objectives, 
and an upper and a lower bound for reporting a certificate for a vertex in the 
winning set. The input is a directed graph $G=(V, E)$ and $k$ Streett 
pairs $( L_\idxs, U_\idxs )$, $\idxs = 1, \dotsc, k$.
The size of the input is measured in terms of $m=\lvert E \rvert$, 
$n=\lvert V \rvert$, $k$, and 
$b = \sum_{\idxs=1}^k (\lvert L_\idxs \rvert + \lvert U_\idxs \rvert) \le 2 n k$.
Our algorithm runs in time $O(n^2 + b \log n)$.

\subsection{Preliminaries}\label{sec:pre}

Let $\OutDeg(G, u)$ be the number of outgoing edges of vertex $u$ in the graph~$G$;
we omit $G$ if clear from the context.
Let $G[S]$ denote the subgraph of a graph $G = (V,E)$ induced 
by the set of vertices $S \subseteq V$. 
$\rev$ denotes the graph with vertices $V$ and all edges of $G$ reversed. 
Let $\reachG{G}{S}$ be the 
set of vertices in $G$ that \emph{can reach} a vertex of $S \subseteq V$. 
A strongly connected component (SCC) of a directed graph $G=(V,E)$ is a subgraph 
$G[S]$ induced by a maximal subset of vertices $S \subseteq V$ such that there
is a path in $G[S]$ between every pair of vertices in $S$. We use the abbreviation 
SCS to denote a strongly connected subgraph that is not necessarily an SCC (i.e.,
is not necessarily maximal w.r.t.\ strong connectivity).
We call an SCS (resp.\ SCC) \emph{trivial} if it only contains a single vertex and no edges.
All other SCSs (resp.\ SCCs) are \emph{non-trivial}.
The set $\reachG{G}{S}$ and the SCCs of a graph $G$ can be found in 
linear time by, e.g., depth-first search~\cite{Tarjan72}.

\smallskip\noindent\emph{Algorithm~\streettalg and good component detection.}
Consider an SCC $\scc$; the \emph{good component detection 
problem} asks to (a) output a non-trivial SCS $G[\gc] \subseteq \scc$ 
induced by some set of vertices~$\gc$
such that for all $1 \le \idxs \le k$ either 
no vertex of $L_\idxs$ or at least one vertex of $U_\idxs$ is contained in the SCS 
(i.e., $L_\idxs \cap \gc=\emptyset$ or $U_\idxs \cap \gc\neq \emptyset$),
or (b) detect that no such SCS exists.
In the former case, there exists an infinite path 
that visits $\gc$ infinitely often and satisfies the Streett objective,
while in the later case there exists no infinite path 
that visits vertices of the SCC~$\scc$ infinitely often and satisfies the Streett objective. 
It follows from the results of~\cite{AlurHenzingerBook} that the following algorithm, 
called Algorithm~\streettalg, suffices for 
the winning set computation: 
\begin{enumerate}%[(1)]
	\item Compute the SCC decomposition of the graph;
 \item for each SCC $\scc$ for which the good component detection 
returns an SCS, label the SCC $\scc$ as \emph{satisfying};
  \item output the set of vertices that can reach a satisfying SCC as the winning set.
\end{enumerate}
Since 
the first and last step are linear time, the running time of Algorithm~\streettalg is 
dominated by the detection of good components in SCCs.
In the following we assume that the input graph is strongly connected and focus
on good component detection.

\subsection{Certificate computation} 
In this subsection we present our results for the certificate computation.
Given a start vertex $x$ that belongs to the winning set, 
a \emph{certificate} or \emph{lasso}~\cite{Ehlers10} is a path from $x$
that consists of a simple path and a not necessarily simple cycle such
that the a play that traverses the cycle infinitely often satisfies the objective.
From the certificate an example of an \emph{accepting run}, i.e., an infinite path 
from $x$ that satisfies the objective, can be constructed.
The output of Algorithm~\streettalg can be used to obtain a certificate.
Let $\gc$ be a set of vertices that induces a good component~$G[\gc]$ and let $x$ 
be a start vertex that can reach $\gc$.
We generate a certificate for $x$ being in the winning set 
as follows. A simple path from $x$ to $\gc$ can be found 
in linear time by a depth-first search. Let $v$ be the vertex of $\gc$ where this path ends. 
We call $v$ the \emph{root} of the good component $G[\gc]$. We show next how to 
obtain, in 
$O(m + n  \min(n, k))$ time, from the good component $G[\gc]$ a cycle starting and ending 
at the root $v$ such that the path resulting from the simple path and the 
cycle is indeed a certificate. For this it is sufficient that the cycle in 
$G[\gc]$ contains for each $L_\idxs$ with $L_\idxs \cap \gc \ne \emptyset$ a vertex 
of $U_\idxs \cap \gc$, i.e., we do not have to include \emph{all} vertices of~$\gc$.

We can use Tarjan's depth-first search based SCC
algorithm~\cite{Tarjan72} to traverse the subgraph $G[\gc]$ in linear $O(m)$ time,
starting from root $v$. Tarjan's algorithm constructs a graph called 
\emph{jungle} with $O(\lvert \gc \rvert)$ edges that for the strongly connected (sub)graph 
$G[\gc]$ consists of a spanning tree and at most one \emph{backedge} per 
vertex of $\gc$. The vertices are assigned pre-order numbers in the order they are traversed. 
We say an edge of $G[\gc]$ is a \emph{backedge} if it leads from a vertex with 
a higher number to a vertex with a lower number. Spanning tree edges always lead
from lower numbered vertices to higher numbered vertices. 
In Tarjan's algorithm a \emph{lowlink} is determined for each vertex $u$ which refers
to the lowest numbered vertex $w$ that $u$ can reach by a sequence of tree edges 
followed by at most one backedge. We additionally store at each vertex $u \ne v$
a \emph{backlink} that is the first edge on the path from $u$ to its lowlink.
The backlinks can be determined and stored during the depth-first search
without increasing its running time.

\begin{example}[Illustration of Tarjan's jungle graph.]
Figure~\ref{fig:tarjan} shows the types of edges and the values
at the vertices as assigned by Tarjan's SCC algorithm for a small example graph.
\end{example}

\begin{figure}
\centering
\begin{tikzpicture}[>=stealth',bend angle=20,font=\footnotesize]
\tikzstyle{vertex}=[circle,draw,thick,inner sep=1pt]
\tikzstyle{arrow}=[->,line width=1pt]
\tikzstyle{label}=[draw,circle,fill,inner sep=1.2pt,solid]
\matrix[column sep=8mm, row sep=8mm]{
	\node[vertex] (1) {$1\,[1]$};
	& \node[vertex] (4) {$4\,[2]$};
	& \node[vertex] (3) {$3\,[2]$};\\
	\node[vertex] (5) {$5\,[1]$};
	& \node[vertex] (2) {$2\,[1]$};
	& \node[vertex] (6) {$6\,[3]$};\\
};
\draw[arrow] (1) -- (2);
\draw[arrow] (2) -- (3);
\draw[decorate,decoration={snake,amplitude=.4mm},arrow] (2) --  (5);
\draw[arrow] (2) -- (6);
\draw[decorate,decoration={snake,amplitude=.4mm},arrow] (3) -- (4);
\draw[decorate,decoration={snake,amplitude=.4mm},arrow,dotted] (4) -- (2);
\draw[decorate,decoration={snake,amplitude=.4mm},arrow,dotted] (5) -- (1);
\draw[decorate,decoration={snake,amplitude=.4mm},arrow,dotted] (6) -- (3);
\end{tikzpicture}
\caption{An example for a ``jungle'' constructed by Tarjan's SCC algorithm
for an SCC.
Backedges are dotted, spanning tree edges are solid. Backlinks are curved. 
The numbers of the vertices represent the order in which the vertices are visited, 
the numbers in brackets are the lowlinks.}\label{fig:tarjan}
\end{figure}

With this data structure we can find within $G[\gc]$ a path from root $v$ to a vertex
$u \in \gc$, $u \ne v$, and back by first searching for $u$ in
the spanning tree and then following the backlinks back to $v$. 
Since no vertex appears more than twice on this path, its size and the 
time to compute it is $O(\lvert \gc \rvert)$.
As it suffices to find such paths for one vertex per non-empty set 
$U_\idxs \cap \gc$, we 
can generate a certificate from $G[\gc]$
in $O(m + \lvert \gc \rvert \min(\lvert \gc \rvert, \lvert\{\idxs \mid 
U_\idxs \cap \gc \ne \emptyset \}\rvert))$ time, which can be bounded by
$O(m + n  \min(n, k))$. This certificate has a size of $O(n \min(n, k))$.

\begin{example}[Lower bound.]\label{ex:lowerbound}
Figure~\ref{fig:certificate} shows that the smallest existing certificate 
can be as large as $\Theta(n  \min(n, k))$.
\end{example}

\begin{figure}
\centering
\begin{tikzpicture}[>=stealth',bend angle=20,font=\small, auto]
\tikzstyle{vertex}=[circle,draw,thick]
\tikzstyle{arrow}=[->,line width=1pt]
\matrix[column sep=8mm, row sep=8mm]{
	\node[vertex, inner sep=2pt] (s) {$s$};
	& & &\\
	\node[vertex,inner sep=2pt] (t) {$t$};
	& \node[vertex,inner sep=1pt] (v1) {$v_1$};
	& \node[vertex,inner sep=1pt] (vi) {$v_i$};
	& \node[vertex,inner sep=1pt] (vk) {$v_{k}$};\\
};
\draw[decorate,decoration={snake},->,line width=1pt] 
	(s) -- node[left] {$\Theta(n)$} (t);
\draw[draw=none] (v1) -- node[auto=false,pos=0.5] {$\cdots$} (vi);
\draw[draw=none] (vi) -- node[auto=false,pos=0.5] {$\cdots$} (vk);
\path
	(t) edge[arrow] (v1)
		edge[arrow, bend right] (vi)
		edge[arrow, bend right] (vk);
\path (v1) edge[arrow, bend right] (s);
\path (vi) edge[arrow, bend right] (s);
\path (vk) edge[arrow, bend right] (s);
\end{tikzpicture}
\caption{Let the only path between $s$ and $t$ be of length $\Theta(n/2) = \Theta(n)$,
not containing any of the vertices $v_\idxs$ for $1 \le \idxs \le k$. Each
$v_\idxs$ has only an edge from $t$ and an edge to $s$.
Let the Streett pairs $( L_\idxs, U_\idxs )$ be given by $L_\idxs = \{s\}$ and
$U_\idxs = \{v_\idxs\}$ for $1 \le \idxs \le k$.
Then the size of the smallest certificate is $\Theta(nk)$, where
$k$ can be of order $n$.}\label{fig:certificate}
\end{figure}

\subsection{Good Component Detection}\label{sec:alg}
In this subsection we present the algorithm for good component detection.
First we introduce the different concepts used in the algorithm for good 
component detection. We start with describing the hierarchical graph decomposition
technique for this setting, which is crucial for the running time 
analysis.

\smallskip\noindent\emph{Graph decomposition.}
In our algorithm we decompose a graph $G$ in the following way.
For $\idxg \in \set{1, \ldots, \lceil \log n \rceil}$, let $G_\idxg = (V, E_\idxg)$ 
be a subgraph of $G$ with $E_\idxg = \{(u,v) \mid \OutDeg(u) 
\le 2^\idxg \}$, i.e., the edges of $G_\idxg$ are the outedges of the vertices with 
outdegree at most $2^\idxg $. 
Note that for $\idxg = \lceil \log n \rceil$ we have that $G_\idxg = G$.
We say vertices in $G$ with $\OutDeg(v) > 2^\idxg $ are \emph{colored blue} in 
$G_\idxg$ and denote the set of blue vertices in $G_\idxg$ by $\blue_\idxg$. 
All other vertices are \emph{white}.
Note that all vertices in $G=G_{\lceil \log n \rceil}$ are white and that all 
vertices in $\blue_\idxg$ have outdegree zero in $G_\idxg$.

\smallskip\noindent\emph{Top and bottom strongly connected components.}
The algorithm repeatedly finds a top or a bottom SCC in the remaining
graph $G$. A bottom SCC $G[S]$ in a directed graph $G$, induced by some set 
of vertices $S$, is an SCC with no edges from vertices in $S$ to vertices 
in $V \setminus S$, i.e., no \emph{outgoing} edges.
A top SCC is a bottom SCC of $\rev$, i.e., an SCC without \emph{incoming}
edges. 
Note that every graph has at least one bottom and at least one top SCC.
If the graph is not strongly connected, then there exist a top and a bottom 
SCC that are disjoint and thus one of them contains at 
most half of the vertices of $G$. 

\smallskip\noindent\emph{Bad vertices.}
In contrast to \emph{good components} we also define \emph{bad vertices}.
The basic idea behind the algorithms for good component detection, described for example
in~\cite{HenzingerT96}, is to repeatedly delete \textit{bad} vertices 
until either a good component is found or it can be concluded that no such component exists. 
A vertex is \emph{bad} if for some index~$\idxs$ with
$1 \le \idxs \le k$ the vertex is in~$L_\idxs$ but it is not strongly connected
to any vertex of~$U_\idxs$. All other vertices are \emph{good}.
Note that good vertices can become bad if some vertex deletion disconnects
an SCS or a vertex of a set $U_\idxs$ is deleted. A good component is
a non-trivial SCS that only contains good vertices.

\smallskip\noindent\emph{Data structure.}
The algorithm maintains for the current graph $G=(V,E)$ (some vertices of the input 
graph might have been deleted) a decomposition into vertex sets $S\subseteq V$ such 
that every SCC of $G$ 
is completely contained in $G[S]$ for one of the sets $S$. For all the sets $S$ a 
data structure $\ds(S)$ is saved in a list $\liste$.
The data structure $\ds(S)$ supports the following operations:
(1) $\construct(S)$ initializes the data structure for the set $S$,
(2) $\remove(S, \ds(S), B)$ removes a set $B \subseteq V$ from $S$ and updates 
the data structure of $S$ accordingly, and
(3) $\bad(\ds(S))$ returns the set $\{v \in S \mid 
\exists \idxs \text{ with } v \in L_\idxs \text{ and } U_\idxs \cap S = \emptyset \}$.
In~\cite{HenzingerT96} an implementation of this data structure is described that achieves
the following running times. For a set of vertices $S \subseteq V$ let $\bits(S)$ 
be defined as $\sum_{\idxs=1}^{k} \left( \lvert S \cap L_\idxs \rvert + \lvert S \cap U_\idxs \rvert\right)$.
\begin{lemma}[Lemma 2.1 in \cite{HenzingerT96}]\label{lemma:ds}
After a one-time preprocessing of time $O(k)$, the data structure $\ds(S)$ can 
be implemented in time $O(\bits(S)+\lvert S \rvert)$ 
for $\construct(S)$, time $O(\bits(B)+\lvert B \rvert)$ for 
$\remove(S, \ds(S), B)$, and constant running time for $\bad(\ds(S))$.
\end{lemma}

\smallskip\noindent\emph{Structure of algorithm.} We denote by $G$ the \emph{current} graph maintained by the 
algorithm where some edges and vertices might have been deleted and use 
\emph{input graph} to denote the unmodified, strongly connected graph for which 
a good component is searched.
Our algorithm for good component detection is given in Algorithm~\ref{alg:streett}. 
It maintains in a list $\liste$ a partition of the vertices in $G$ into sets 
such that every SCC of $G$ is contained in the subgraph induced by one of the vertex 
sets. The list is 
initialized with the set of all vertices in the strongly connected input graph.
We show that if a good component exists, its vertices must be fully 
contained in one of 
the vertex sets in the partition. The algorithm repeatedly removes a set $S$ from 
$\liste$ and identifies and deletes bad vertices from $G[S]$.
If no edge is contained in $G[S]$, the set $S$ is removed as
it can only induce trivial SCCs.
Otherwise the subgraph $G[S]$ is either 
determined to be strongly connected and output as a good component or
a ``small'' SCC in $G[S]$ is identified.

\smallskip\noindent\emph{Search for small SCCs.}
To find a small SCC, the algorithm searches alternatingly in $G[S]$ 
and in $\rev[S]$ for a bottom SCC and stops as soon as one of the searches stops. 
(A bottom SCC in $\rev[S]$ is a top SCC in $G[S]$.) 
We only describe the search in $G[S]$ here, the search in $\rev[S]$ is analogous.
The algorithm uses the hierarchical graph decomposition of $G[S]$.
The subgraph $G_\idxg[S]$  for any $\idxg$ contains only the outedges
of vertices with an outdegree of at most $2^\idxg$. The search for a
bottom SCC is started at $\idxg = 1$, then $\idxg$ is increased one by one
if necessary, up to at most $\lceil \log (\lvert S \rvert) \rceil$.
If  for some $\idxg$ we can identify a bottom SCC that 
does not contain any blue vertex (i.e., a vertex for which some edges are 
missing in $G_\idxg$), then the found SCC in $G_\idxg[S]$ must also be a bottom 
SCC in $G[S]$. If multiple bottom SCCs (without blue vertices) are found in 
$G_\idxg[S]$, we only consider the smallest one. The Procedure~$\bsccalg(H')$ returns
the set of vertices that induces the smallest bottom SCC in the graph $H'$.
We then put the newly detected SCC and the ``rest'' of $S$ back into $\liste$.

\begin{algorithm2e}
\SetAlgoRefName{GoodComp}
\caption{Detection of good components for the winning set computation in graphs 
with $k$-pair Streett objectives}
\label{alg:streett}
\SetKwInOut{Input}{Input}
\SetKwInOut{Output}{Output}
\BlankLine
\Input{\emph{strongly connected graph} $G = (V, E)$ and \\
\emph{Streett pairs} $( L_\idxs, U_\idxs )$ for $\idxs = 1, \dotsc, k$}
\Output{a good component in $G$ if one exists}
\BlankLine
add $\construct(V)$ to $\liste$\;
\While{$\liste \ne \emptyset$}{ \label{l:while}
	pull $\ds(S)$ from $\liste$\label{l:pull}\;
	\While{$\bad(\ds(S)) \ne \emptyset$}{\label{l:badstart}
		$\ds(S) \leftarrow \remove(S, \ds(S), \bad(\ds(S)))$\label{l:badend}\;
	}
	\If{$G[S]$ contains at least one edge}{ \label{l:le2}
		\For{$\idxg \leftarrow 1$ \KwTo $\lceil \log (\lvert S \rvert) \rceil$}{\label{l:for}
			\ForEach{$H \in \set{G,\rev}$}{
				construct $H_\idxg[S]$\label{l:forl1}\;
				$\blue_\idxg \leftarrow \{v \in S \mid \OutDeg(H, v) > 2^\idxg\}$\;
				$Z \leftarrow S \setminus \reachG{H_\idxg[S]}{\blue_\idxg}$\;
				\tcc*[l]{$Z$ cannot reach $\blue_\idxg$} %\;
				\If{$Z \ne \emptyset$}{
					$\gc \leftarrow \bsccalg(H_\idxg[Z])$\label{l:scc}\;
					\If{$\gc = S$}{
						\tcc*[l]{good component found}
						\KwRet{$G[S]$}\label{l:stillconnected}\;
					}
					\If{$\lvert \gc \rvert \le \lvert S \rvert / 2$}{\label{l:smallerhalf}
						add $\remove(S, \ds(S), \gc)$ to $\liste$\label{l:addRest}\;
						add $\construct(\gc)$ to $\liste$\label{l:addX}\;
						continue with pull of next $\ds(S)$ from $\liste$ (Line~\ref{l:pull})\;
					}
				}
			}
		}
	}
}
\KwRet{no good component exists}
\end{algorithm2e}

\smallskip\noindent\emph{Outline running time.}
The idea of the running time analysis is as follows.
We can show that a bottom SCC of $G[S]$ identified in iteration~$\idxg$ of the 
outer for-loop 
must contain $\Omega(2^\idxg)$ vertices. In
time $O(n \cdot 2^\idxg)$ a standard SCC algorithm can compute all SCCs of $G_\idxg[S]$
and thus also the smallest bottom SCC. The time needed for
the search in all graphs $G_{{\idxg}'}[S]$ for $1 \le {\idxg}' < \idxg$ can be 
bounded by an additional factor
of two. Thus the work for the search is $O(n)$ per vertex in the identified SCC.

Given that the subgraph $G[S]$ was split into at least one top and one bottom SCC, 
the smallest top or bottom SCC contains at most half of the vertices of the 
subgraph. By searching for a smallest bottom SCC (without blue vertices) 
in $G_\idxg[S]$ and $\rev_\idxg[S]$, we find one top or bottom SCC with at most 
half of the vertices of the subgraph. We charge the work for finding
such an SCC to the vertices in this SCC. We show that this yields a 
total running time of~$O(n^2)$ for computing SCCs.

We additionally have to take the time for the maintenance of the data structures 
into account. Here we use the properties of the data structure $\ds(S)$ described
in Lemma~\ref{lemma:ds} to obtain a running time of $O((n + b) \log n)$ for the
maintenance of the data structures and the identification of bad vertices over the
whole algorithm. 
Combined these ideas lead to a total running time of $O(n^2 + b \log n)$.

\begin{lemma}\label{lemma:size}
Let $H \in \set{G,\rev}$ be the graph and let ${\idxg}^*$ be the iteration
for which in Algorithm~\ref{alg:streett} the outer for-loop stops.
Let $Z$ be the non-empty set $S \setminus \reachG{ 
H_{{\idxg}^*}[S]}{\blue_{{\idxg}^*}}$ and let $\gc$ be the set of vertices that induces the 
smallest bottom SCC $H[\gc]$ in $H_{{\idxg}^*}[Z]$ returned by 
$\bsccalg(H_{{\idxg}^*}[Z])$. Assume we have $ \lvert \gc \rvert \le \lvert S 
\rvert /2$. Then $H[\gc]$ contains at least $2^{{\idxg}^*-1}$ vertices.
\end{lemma}

\begin{proof}
As $\blue_{{\idxg}^*-1}$ is the set of vertices in $H_{{\idxg}^*-1}[S]$ with 
outdegree larger than $2^{{\idxg}^*-1}$, any bottom SCC $H[Y]$ that contains a vertex of $\blue_{{\idxg}^*-1}$, 
has $\lvert Y \rvert \ge 2^{{\idxg}^*-1}$.
Hence it suffices to show that $\gc \cap \blue_{{\idxg}^*-1} \ne \emptyset$.
Assume by contradiction that  $\gc \cap \blue_{{\idxg}^*-1} = \emptyset$.
Since $H[\gc]$ is a bottom SCC, % and $Y \cap \blue_{{\idxg}^*-1} = \emptyset$, 
no vertex
of $\gc$ can reach any vertex of $\blue_{{\idxg}^*-1}$, i.e.,
$\gc \subseteq S \setminus \reachG{H_{{\idxg}^*}[S]}{\blue_{{\idxg}^*-1}}$. As all edges in $H_{{\idxg}^*-1}[S]$
are contained in $H_{{\idxg}^*}[S]$, this implies 
$\gc \subseteq S \setminus \reachG{H_{{\idxg}^*-1}[S]}{\blue_{{\idxg}^*-1}}$.
Since $\bsccalg$ finds the smallest bottom SCC in graph $H_\idxg$ for each $\idxg$, the outer for-loop would thus
have terminated in an iteration $\idxg \le {\idxg}^*-1$. Contradiction. 
\end{proof}

\begin{lemma}[Running time]\label{lemm:timeStreett}
Algorithm~\ref{alg:streett} can be implemented in time $O(n^2 + b \log n)$.
\end{lemma}

\begin{proof}
The preprocessing and initialization of the data structure and the removal 
of bad vertices in the whole algorithm take time $O(m + k + b)$ using
Lemma~\ref{lemma:ds}.
Additionally we maintain 
at each vertex a list of its incoming and a list of its outgoing edges 
including pointers to the lists of its neighbors, which we
use  to update the lists of its neighbors. Since each vertex is
deleted at most once, this data structure can be constructed and maintained in 
total time $O(m)$.

Consider the while loop where a set $S$ is removed from $\liste$.
If $G[S]$ does not contain any edge after the removal of bad vertices,
then $S$ is not considered further by the algorithm. Otherwise $G[S]$
and $\rev[S]$ are search for bottom SCCs.
The search in $G[S]$ and $\rev[S]$ only increases the running time 
by a factor of two, thus we restrict the analysis of the running time to $G[S]$.
Let $n' \le n$ be the number of vertices in~$S$.
The construction of $ G_\idxg[S]$, $Z$, and $G[\gc]$ can all be done in time 
$O(n' \cdot 2^\idxg)$ for each~$\idxg$, i.e., in total time 
$O(n' \cdot  2^{{\idxg}^*})$ up to 
level~$\idxg^*$. If $\gc=S$, then the algorithm terminates
and the time for processing $S$ can be bounded by $O(n' \cdot 2^{\log n'}) = O((n')^2)$.
If the processing of~$S$ ends when some bottom SCC $G[\gc] \subseteq G[S]$ 
induced by some set of vertices $\gc$ is found,
let $\idxg^*$ be the value of $\idxg$ when $G[\gc]$ is detected and inserted 
into $\liste$, and let $c$ be some constant such that the time spent in this 
search for $\gc$ is bounded by $c \cdot n' \cdot 2^{{\idxg}^*-1}$.
By Lemma~\ref{lemma:size} the set~$\gc$ contains at least $2^{{\idxg}^*-1}$ vertices.
Let $\lvert \gc \rvert = n_1$. The algorithm ensures $n_1 \le n' / 2$. 
We claim that the total running time for processing all sets $S$, except for 
the work in $\remove$ and $\construct$, can be bounded by $f(n) = 2 c n^2$.
Whenever the algorithm does not terminate, we have by induction,
and in particular for $n' = n$,
\begin{align*}
	f(n') &\le f(n_1) + f(n' - n_1) + c n' n_1 \,,\\
	&\le 2c n_1^2 + 2c (n'-n_1)^2 + c n' n_1 \,,\\
	&= 2c n_1^2 + 2c (n')^2 - 4c n' n_1 + 2c n_1^2 + c n' n_1 \,,\\
	&= 2c (n')^2 + 4c n_1^2 - 3c n' n_1 \,,\\
	&\le 2c (n')^2 \,,
\end{align*}
where the last inequality follows from $n_1 \le n'/2$.

The operations $\remove$ and $\construct$ are called once per found bottom SCC 
$G[\gc]$ with $\gc \ne S$ and take by Lemma~\ref{lemma:ds} $O(\lvert \gc \rvert +
\bits(\gc))$ time. By $n_1 \le n' / 2$ we have that whenever a vertex~$v$ is 
in~$\gc$, the size of the set in $\liste$ containing $v$ is halved; this can 
happen at most $\lceil \log n \rceil$ times.
Hence, by charging $O(1)$ to the vertices in $\gc$ and, respectively, to $\bits(\gc)$, 
the total running time for this part can be bounded by $O((n + b) \log n)$, as 
each vertex and bit is only charged $O(\log n)$ times.
Combining all parts yields the claimed running time bound of $O(n^2 + b \log n)$.
\end{proof}

To prove the correctness of Algorithm~\ref{alg:streett},
we first show that all candidates for good components are in~$\liste$ before
each iteration of the algorithm.
\begin{lemma}\label{lemma:allcandidates}
Before each iteration of the outer while-loop every good component of the input
graph is contained in one of the subgraphs~$G[S]$ for which the data structure
$\ds(S)$ is maintained in the list~$\liste$.
\end{lemma}
\begin{proof}
We show that the algorithm never removes edges or vertices that belong to
a good component, which together with a correct initialization of the list~$\liste$
implies the lemma.
At the beginning of the algorithm one data structure for the whole strongly connected
input graph is added to~$\liste$. Thus every good component is contained in
this data structure in~$\liste$ after the initialization. 
At the beginning of each iteration of 
the outer while-loop the data structure of one of the subgraphs $G[S]$
is pulled from the list~$\liste$. In Lines~\ref{l:badstart}--\ref{l:badend} we 
remove vertices from the subgraph that are in some set $L_\idxs$ but not 
strongly connected to any vertex in $U_\idxs$, i.e., bad vertices.
In Line~\ref{l:le2} we remove trivial SCCs.
Observe that a good component is non-trivial and does not
contain any bad vertices. Thus the removal of bad vertices
and trivial SCCs does not remove any vertices of a good component,
i.e., after the removal of these vertices
the updated subgraph $G[S]$ still contains the good components it contained 
before. If no good component is identified in this iteration, i.e., the algorithm
does not terminate, we find a bottom or top SCC $G[\gc]$, induced by some
set of vertices~$\gc$. Since a good component is strongly connected, every good 
component in $G[S]$ either is a subgraph of the newly identified SCC $G[\gc]$
or does not contain \emph{any} vertex of~$\gc$. Thus the removed edges between $G[\gc]$ 
and the remaining subgraph cannot belong to a good component. Finally, we add the 
data structures for $G[\gc]$ as well as for $G[S \setminus \gc]$ to $\liste$.
Thus no vertex or edge of a good component was removed and every good component
continues to be completely contained in a subgraph in~$\liste$.
\end{proof}

As all candidates for good components are maintained in the list $\liste$,
it remains to show that the algorithm  
correctly outputs a good component if and only if one exists. 
\begin{lemma}[Correctness] \label{lemm:correctStreett}
Algorithm~\ref{alg:streett} outputs a good component if one exists,
otherwise the algorithm reports that no such component exists.
\end{lemma}

\begin{proof}
First we show that whenever Algorithm~\ref{alg:streett} outputs a subgraph $G[S]$
induced by some set of vertices~$S$, then $G[S]$ is a good component.
Line~\ref{l:le2} ensures only non-trivial SCSs are considered. 
After the removal of bad vertices from~$S$ in Lines~\ref{l:badstart}--\ref{l:badend},
we know that for 
all $1 \le \idxs \le k$ and all vertices in $S \cap L_\idxs$ there exists a vertex in 
$S \cap U_\idxs$. Thus if $G[S]$ is strongly connected, then $G[S]$ is a good 
component, and the only SCC in $G[S]$ is $G[S]$ itself. Only in this case $G[S]$ 
is output (in Line~\ref{l:stillconnected}).

Algorithm~\ref{alg:streett} terminates if a good component is identified or
$\liste$ is empty; in the latter case it reports that 
no good component exists.
Lemma~\ref{lemma:allcandidates} shows that before every 
iteration of the outer while-loop \emph{every} good component is contained in 
one of the subgraphs $G[S]$ in $\liste$. That is, if a good component exists
in $G$ and no good component was identified yet by the algorithm, then 
$\liste$ is not empty and
thus the algorithm does not terminate until a good component is identified.
Hence if the algorithm terminates because $\liste$ is empty, then no 
good component exists.
By Lemma~\ref{lemm:timeStreett} the algorithm terminates after a finite number 
of steps.

Next we show that if there exists a good component in $G$, then 
the algorithm outputs a good component.
Let~$Y$ be a {\em maximal} good component in~$G$ 
and let~$S_Y$ be the vertex set maintained in $\liste$
that currently contains the vertices in~$Y$. By the arguments above after a 
finite number of steps either (1) another good component is detected or
(2) $\ds(S_Y)$ is pulled from $\liste$. In Case~(1) we are done, the argument 
for Case~(2) is as follows. 
By Lemma~\ref{lemma:allcandidates} the component~$Y$ is never split
by the algorithm thus after Case~(2) happened at most $n$ times, one of 
the following two cases occurs: either (a) $\ds(S_Y)$ is pulled from~$\liste$
with $G[S_Y] \supset Y$ and after the removal of
bad vertices from $S_Y$, $G[S_Y]$ without the bad vertices is equal to $Y$ or (b) $G[S_Y] = Y$ is pulled 
from $\liste$. In both cases the good component $Y$ is output and the algorithm
terminates: Since $Y$ is non-trivial, the condition in Line~\ref{l:le2} is satisfied.
The algorithm searches for a top or bottom SCC in~$Y$. Since~$Y$ is 
strongly connected, the only top or bottom SCC in~$Y$ is~$Y$ itself. Hence 
the algorithm outputs~$Y$ in Line~\ref{l:stillconnected}. 
\end{proof}

Recall Algorithm~\streettalg that calls Algorithm~\ref{alg:streett} for each
SCC in the input graph and then computes reachability to 
the union of the identified good components. Lemmata~\ref{lemm:timeStreett} 
and~\ref{lemm:correctStreett} yield the following result.

\begin{theorem}
Algorithm~\streettalg correctly computes the winning set in graphs with $k$-pair Streett
objectives in $O(n^2+b\log n)$ time. Given a vertex $x$ in the winning set
and a good component reachable by $x$, a certificate for $x$ can be output in 
time $O(m+ n \min(n,k))$.
\end{theorem}

\begin{remark}[Optimality]
We have shown that in a graph with $k$-pair Streett objectives the winning set and 
a certificate can be computed and output in time $O(n^2+b\log n)$.
Example~\ref{ex:lowerbound} shows a lower bound of $\Omega(n  \min(n, k))$ 
for outputting a certificate. Note that the size of the input is at least $b$.
Hence the presented algorithm is optimal up to a 
log factor when $k=\Omega(n)$ and a certificate is required.
\end{remark}

\section{Conclusion}
In this work we have considered two classical algorithmic questions for parity and 
Streett objectives on graphs and game graphs. 

We have presented an algorithm for parity games with $n$ vertices and $\numprio$
priorities with a running time of roughly $O(n^{4/3+\numprio/3})$, which improves
the running time over previous results for game graphs with
$\Omega(n^{4/3})$ edges when the number of priorities is sub-polynomial in $n$.
In particular we improved the long standing running time for 3~priorities from 
$O(mn)$ to $O(n^{5/2})$.

For graphs with Streett objectives we have presented an $
O(n^2 + b \log n)$-time algorithm (where $b$ is the total number of elements in 
the Streett pairs), and a lower bound 
that shows that this running time is tight up to a log factor when a certificate
has to be reported. The algorithm improves upon the known running time bounds 
when the number of edges~$m$ is at least of order $n^{4/3} \log^{-1/3} n 
+ b^{2/3} \log^{1/3} n$ and the number of Streett pairs is at least of order 
$n^2/m$.
This algorithm was extended to Markov Decision Processes (MDPs) in subsequent work~\cite{ChatterjeeDHL16,Loitzenbauer16}.

In particular in the light of the quasi-polynomial time algorithm for parity 
games in a breakthrough result subsequent to our work~\cite{CaludeJKLS17}, 
showing new upper and (conditional) lower bounds for parity games remains
a very interesting challenge. The techniques of \cite{CaludeJKLS17} as well as the 
techniques presented here are inherently non-symbolic. An interesting open 
question is thus to find improved symbolic algorithms for these classical problems.

\section*{Acknowledgement}
\noindent We would like to thank Tom Henzinger and the anonymous 
referees for their useful comments.

%The authors are partially supported by the Vienna Science and Technology Fund 
%(WWTF) grant ICT15-003 and the Austrian Science Fund (FWF): P23499-N23.
%K.~C.\ is partially 
%supported by the Austrian Science Fund (FWF): S11407-N23 (RiSE/SHiNE), 
%an ERC Start Grant (279307: Graph Games), and a Microsoft Faculty Fellows Award.
%For M.~H.\ and V.~L.\ 
%the research leading to these results has received funding from the European 
%Research Council under the European Union's Seventh Framework Programme 
%(FP/2007-2013) / ERC Grant Agreement no. 340506 and the Vienna Science and 
%Technology Fund (WWTF) grant ICT10-002. 

\bibliographystyle{plain}
\bibliography{../literature}

\vspace{-20 pt}
\end{document}